\documentclass[reqno]{amsart}
\usepackage{amsmath}
\usepackage{amsthm}
\usepackage{graphicx}
\usepackage{amsfonts}
\usepackage{amssymb}
\usepackage{enumerate}
\usepackage[top=1in, bottom=1in, left=1in, right=1in]{geometry}

\numberwithin{equation}{section}

\newtheorem{theorem}{Theorem}[section]
\newtheorem{lemma}[theorem]{Lemma}
\newtheorem{proposition}[theorem]{Proposition}
\newtheorem{corollary}[theorem]{Corollary}

\theoremstyle{definition}

\newtheorem{definition}[theorem]{Definition}

\newtheorem{remark}[theorem]{Remark}

\def\E{{\mathbb E}}
\def\R{{\mathbb R}}

\def\P{{\mathcal P}}

\def\Q{{\mathcal Q}}

\def\G{{\mathcal G}}

\def\A{{\mathcal A}}
\def\F{{\mathcal F}}

\DeclareMathOperator*{\esssup}{ess\,sup}

\title{Law invariant risk measures and information divergences}
\author{Daniel Lacker}

\thanks{This material is based upon work supported by the National Science Foundation under Award No. DMS-1502980}

\address{\noindent Brown University, Division of Applied Mathematics, 182 George St, Providence, RI 02906}
\email{daniel\_lacker@brown.edu}

\begin{document}

\begin{abstract}
A one-to-one correspondence is drawn between law invariant risk measures and divergences, which we define as functionals of pairs of probability measures on arbitrary standard Borel spaces satisfying a few natural properties. Divergences include many classical information divergence measures, such as relative entropy and $f$-divergences. Several properties of divergence and their duality with law invariant risk measures are developed, most notably relating their chain rules or additivity properties with certain notions of time consistency for dynamic law invariant risk measures known as acceptance and rejection consistency. These properties are linked also to a peculiar property of the acceptance sets on the level of distributions, analogous to results of Weber on weak acceptance and rejection consistency. Finally, the examples of shortfall risk measures and optimized certainty equivalents are discussed in some detail, and it is shown that the relative entropy is essentially the only divergence satisfying the chain rule.
\end{abstract}

\maketitle

\section{Introduction}

This paper deepens the analysis of law invariant risk measures and their connection to divergence-type functionals of probability measures. Throughout the paper, a nonatomic standard Borel space $(\Omega,\F,P)$ is fixed, and a \emph{risk measure} is defined to be a convex functional $\rho : L^\infty := L^\infty(\Omega,\F,P) \rightarrow \R$ satisfying:
\begin{enumerate}
\item Monotonicity: If $X,Y \in L^\infty$ and $X \le Y$ a.s. then $\rho(X) \le \rho(Y)$.
\item Cash additivity: If $X \in L^\infty$ and $c \in \R$ then $\rho(X + c) = \rho(X) + c$.
\item Normalization: $\rho(0)=0$.
\end{enumerate}
The functional $X \mapsto \rho(-X)$ is more traditionally called a \emph{normalized convex risk measure}; some authors use the term \emph{acceptability measure} \cite{roorda2007time} for what we have chosen to call a risk measure. Convex risk measures first appeared in \cite{follmer-schied-convex,frittelli2002putting,heath2004pareto}, extending the class of \emph{coherent} risk measures introduced in the seminal paper of Artzner et al. \cite{artzner1999coherent} (see also \cite{delbaen2002coherent}). A risk measure $\rho$ is \emph{law invariant} if $\rho(X)=\rho(Y)$ whenever $X$ and $Y$ have the same law. Three standard examples will guide us throughout the paper: The first is the well known entropic risk measure $\rho(X) = \eta^{-1}\log\E[e^{\eta X}]$, $\eta > 0$. Second, given a nondecreasing convex function $\ell : \R \rightarrow [0,\infty)$ with $\ell(0)=1$, the corresponding \emph{shortfall risk measure} (introduced by F\"ollmer and Schied in \cite{follmer-schied-convex}) is
\[
\rho(X) = \inf\left\{c \in \R : \E[\ell(X-c)] \le 1\right\}.
\]
Lastly, given a nondecreasing convex function $\phi : \R \rightarrow \R$ with $\phi^*(1) = \sup_{x \in \R}(x-\phi(x)) = 0$, the corresponding
\emph{optimized certainty equivalent} (introduced by Ben-Tal and Teboulle in \cite{bental-teboulle-1986,bental-teboulle-2007}) is
\[
\rho(X) := \inf_{m \in \R}\left(\E[\phi(m+X)] - m\right).
\]

We construct divergences as follows:
Fix a law invariant risk measure $\rho$. 
Given a Polish space $E$, let $\P(E)$ denote the set of Borel probability measures on $E$. For any Polish space (or any standard Borel space) $E$ and any $\mu \in \P(E)$, we may define a new law invariant risk measure $\rho_\mu : L^\infty(E,\mu) \rightarrow \R$ by $\rho_\mu(f) := \rho(f(X))$, where $X$ is any $E$-valued random variable on $\Omega$ with law $P \circ X^{-1} = \mu$. Indeed, such an $X$ exists because $\Omega$ is nonatomic, and this definition is independent of the choice $X$ thanks to law invariance.
This family of risk measures satisfies a consistency property, namely
\begin{align}
\rho_\mu(f) = \rho_\nu(g), \text{ whenever } \mu \circ f^{-1} = \nu \circ g^{-1}. \label{def:riskmeasureconsistency}
\end{align}
Let $\alpha(\cdot | \mu)$ denote the minimal penalty function associated to $\rho_\mu$, i.e., the restriction to $\P(E)$ of the convex conjugate of $\rho_\mu$:
\[
\alpha(\nu | \mu) = \sup_{f \in L^\infty(E,\mu)}\left(\int_Ef\,d\nu - \rho_\mu(f)\right) = \sup\left\{\int_Ef\,d\nu : f \in L^\infty(E,\mu), \ \rho_\mu(f) \le 0\right\}.
\]
We call $\alpha$ the \emph{divergence induced by $\rho$}.
In summary, the functional $\alpha(\cdot | \cdot)$ is defined for pairs of probability measures on \emph{any Polish space} (or standard Borel space), much like the classical relative entropy and other information divergences, such as the $f$-divergence.
Indeed, when $\rho$ is the entropic risk measure, $\alpha$ is nothing but the usual relative entropy (also known as the Kullback-Leibler divergence)
\[
H(\nu | \mu) = \int \log\left(\frac{d\nu}{d\mu}\right)\,d\nu \ \ \text{ for } \nu \ll \mu, \quad \infty \text{ otherwise}.
\]
When $\rho$ is a shortfall risk measure corresponding to a function $\ell$, the induced divergence is
\[
\alpha(\nu | \mu) = \inf_{t > 0}\frac{1}{t}\left(1 + \int_E\ell^*\left(t\frac{d\nu}{d\mu}\right)d\mu\right), \text{ for } \nu \ll \mu, \quad \infty \text{ otherwise},
\]
where $\ell^*(t) = \sup_{s \in \R}(st - \ell(s))$ is the convex conjugate of $\ell$.
Finally, when $\rho$ is the optimized certainty equivalent corresponding to a function $\phi$, the induced divergence is the $\phi^*$-divergence
\[
\alpha(\nu|\mu) = \int\phi^*\left(\frac{d\nu}{d\mu}\right)d\mu, \text{ for } \nu \ll \mu, \quad \infty \text{ otherwise}.
\]
In fact, we could instead start from a $[0,\infty]$-valued function $\alpha=\alpha(\nu|\mu)$, defined for pairs of probability measure $(\nu,\mu) \in \P(E)^2$ for any Polish space $E$, such that for each Polish space $E$ and each $\mu \in \P(E)$ we have the following properties:
\begin{enumerate}
\item $\alpha(\mu | \mu) = 0$.
\item $\alpha(\nu | \mu) = \infty$ if $\nu \in \P(E)$ is not absolutely continuous with respect to $\mu$.
\item The map $\nu \mapsto \alpha(\nu | \mu)$ is convex and lower semicontinuous with respect to total variation.
\item $\alpha(\nu K | \mu K)\le \alpha(\nu | \mu)$ for every $\nu \in \P(E)$ and every kernel $K$ from $E$ to another Polish space $F$, where $\mu K(dy) := \int_E\mu(dx)K(x,dy) \in \P(F)$.
\end{enumerate}
We call such a functional a \emph{divergence},
and we show that to any divergence there corresponds a unique law invariant risk measure defined on the original space $(\Omega,\F,P)$; we prove this by showing the definitions
\[
\rho(f(X)) := \rho_\mu(f) := \sup_{\nu \in \P(E)}\left(\int_Ef\,d\nu - \alpha(\nu | \mu)\right),
\]
to be consistent in the sense of \eqref{def:riskmeasureconsistency}, where $E$ is a Polish space, $f \in B(E)$, and $\mu = P \circ X^{-1}$ for some $X : \Omega \rightarrow E$. The property (4) corresponds exactly to the consistency property \eqref{def:riskmeasureconsistency} and is known as the \emph{data processing inequality} in information theory, at least when $\alpha$ is the usual relative entropy.

The primary focus of the paper is on the characterization of properties of divergences related to the well known chain rule for relative entropy, which reads
\[
H\left(\nu(dx)K^\nu_x(dy) \ | \ \mu(dx)K^\mu_x(dy)\right) = H(\nu | \mu) + \int\nu(dx)H(K^\nu_x | K^\mu_x),
\]
and holds for all (disintegrated) probability measures $\mu(dx)K^\mu_x(dy)$ and $\nu(dx)K^\nu_x(dy)$ on the product of two Polish spaces. More generally, we say a divergence $\alpha$ is \emph{superadditive} if
\begin{align}
\alpha\left(\nu(dx)K^\nu_x(dy) \ | \ \mu(dx)K^\mu_x(dy)\right) \ge \alpha(\nu | \mu) + \int\nu(dx)\alpha(K^\nu_x | K^\mu_x). \label{intro:superadditivity}
\end{align}
and we say $\alpha$ is \emph{subadditive} if the reverse inequality holds. We characterize this in terms of various properties of the corresponding risk measure $\rho$.

The original motivation for this study comes from an ongoing investigation into the tensorization properties of \emph{concentration inequalities} of the form 
\begin{align}
\rho(\lambda X) \le \gamma(\lambda), \text{ for all } \lambda \ge 0, \label{intro:concentration}
\end{align}
where $\gamma : [0,\infty) \rightarrow [0,\infty]$. In a follow-up paper \cite{lacker-liquidity}, we study concentration inequalities \eqref{intro:concentration} in connection with liquidity risk. When $\rho$ is the entropic risk measure, the inequality \eqref{intro:concentration} is simply a bound on the moment generating function of $X$. Tensorization in this context roughly means bounding $\rho(\lambda h(X,Y))_{\lambda \ge 0}$ in terms of bounds on $\rho(\lambda f(X))_{\lambda \ge 0}$ and $\rho(\lambda g(Y))_{\lambda \ge 0}$, for two given (typically independent) random variables $X$ and $Y$ and various (classes of) functions $f,g,h$. Tensorization properties are typically proven using the chain rule (see \cite{gozlan-leonard-survey} for details, particularly Proposition 1. thereof), so we seek alternatives for the chain rule in order to understand how to extend these ideas to general concentration inequalities of the form \eqref{intro:concentration}.

It turns out that the dual form of superadditivity \eqref{intro:superadditivity} is a so-called \emph{time-consistency} property of the corresponding risk measure $\rho$, which we describe by building on a construction of Weber \cite{weber-distributioninvariant}: Define a functional $\tilde{\rho}$ on $\P(\R)$ by $\tilde{\rho}(P \circ X^{-1}) = \rho(X)$, which is of course well defined thanks to law invariance. For any $\sigma$-field $\G \subset \F$ in $\Omega$ and any $X \in L^\infty$, consider the $\G$-measurable random variable
\[
\rho(X | \G)(\omega) := \tilde{\rho}(P(X \in \cdot \, | \, \G)(\omega)),
\]
where $P(X \in \cdot \, | \, \G)$ denotes a regular conditional law of $X$ given $\G$. We say $\rho$ is acceptance consistent if $\rho(X) \le \rho(\rho(X | \G))$ for every $X \in L^\infty$ and any $\sigma$-field $\G \subset \F$.
If the reverse inequality holds, we say $\rho$ is \emph{rejection consistent}. If $\rho$ is both acceptance and rejection consistent, we say it is \emph{time consistent}. We show that acceptance consistency of $\rho$ is essentially equivalent to the superadditivity of the induced divergence $\alpha$, and we provide an additional characterization in terms of a property of the \emph{measure acceptance set}
\[
\A := \left\{P \circ X^{-1} : X \in L^\infty, \ \rho(X) \le 0 \right\} \subset \P(\R).
\]
These various characterizations are put to use to find those functions $\ell$ and $\phi$ for which the corresponding shortfall risk measures and optimized certainty equivalents are acceptance consistent.
It follows from the results of Kupper and Schachermayer \cite{kupper-schachermayer} that the entropic risk measure is essentially the only time consistent risk measure, and as a corollary we find that the relative entropy is the only divergence (up to a scalar multiple) satisfying the chain rule.\footnote{We make no attempt to reconcile our characterization of relative entropy with the many already present in the literature (see the survey of Csisz\'ar \cite{csiszar2008axiomatic}), but we can at least say with confidence that the techniques by which we obtained it are new, notably avoiding functional equations.} Ultimately, we find that not many law invariant risk measures are acceptance consistent (or rejection consistent) other than the entropic one, or modest perturbations thereof. In other words, not many divergences beyond relative entropy are superadditive. Although our results are somewhat negative in this sense, the construction and characterization of divergences induced by risk measures is interesting in its own right, and they appear to be useful tools in the study of law invariant risk measures. Moreover, we find some value in understanding the limitations of our divergences in the applications discussed above.

We also briefly revisit the related results of Weber \cite{weber-distributioninvariant}. Say that $\rho$ is \emph{weakly acceptance consistent} if $\rho(X) \le 0$ whenever $\rho(X | \G) \le 0$ a.s., for $X \in L^\infty$ and $\sigma$-fields $\G \subset \F$. Weber showed that this is essentially equivalent to the convexity of the measure acceptance set $\A$.
We show that weak acceptance consistency is also equivalent to an inequality weaker than superadditivity:
\[
\alpha\left(\nu(dx)K^\nu_x(dy) \ | \ \mu(dx)K^\mu_x(dy)\right) \ge \int\nu(dx)\alpha(K^\nu_x | K^\mu_x).
\]
Time consistency properties of dynamic risk measures have by now been studied thoroughly \cite{riedel2004dynamic,detlefsen2005conditional,frittelli2004dynamic,follmer2006convex,tutsch2008update,cheridito2011composition}. The nice survey of Acciaio and Penner \cite{acciaio-penner-dynamic} will be a useful reference, although we will mostly work with the type of dynamic law invariant risk measures constructed by Weber in \cite{weber-distributioninvariant}. With this rich literature in mind, the most novel of our results on time consistency is the characterization of acceptance consistency in terms of the shift-convexity of the measure acceptance set, which nicely complements Weber's result on weak acceptance consistency. In theory, our characterization in terms of superadditivity \eqref{intro:superadditivity} could be deduced from results in \cite{acciaio-penner-dynamic}, but this is non-trivial: The key difference is that previous papers on the subject (including \cite{acciaio-penner-dynamic}) use \emph{essential suprema} to define the minimal penalty function of a conditional risk measure. We work purely with pointwise definitions, and while this distinction is largely technical, there is a non-trivial gap between the two stemming from a delicate measurable selection argument. See Section \ref{se:essentialsuprema} for details.

The above results must be qualified: the equivalence of superadditivity and acceptance consistency is only proven under the additional assumption that the divergence $\alpha$ is \emph{simplified}, in the sense that
\[
\alpha(\nu | \mu) = \sup_{f \in C([0,1])}\left(\int_{[0,1]}f\,d\nu - \rho_\mu(f)\right),
\]
for each $\mu,\nu \in \P([0,1])$. This additional assumption is admittedly somewhat of a nuisance, and it is unclear if the main result on superadditivity holds without it. While we did not identify a nice dual characterization for this condition, we have identified a common stronger condition: Namely, $\alpha$ is simplified as soon as $\rho$ is \emph{Lebesgue continuous}, in the sense that $\rho(X_n) \rightarrow \rho(X)$ whenever $X_n$ are uniformly bounded and $X_n \rightarrow X$ a.s. This is a strong assumption, but it indeed holds for our main examples of shortfall risk measures and optimized certainty equivalents. Moreover, we show that Lebesgue continuity of $\rho$ is actually equivalent to joint lower semicontinuity of the induced divergence $\alpha$ (with respect to weak convergence).

Finally, Section \ref{se:furtherproperties} we study miscellaneous properties of divergences.
First, joint convexity of $\alpha$ is shown to be equivalent to the concavity of $\rho$ on the level of probability measures (i.e., concavity of $\tilde{\rho}$ defined above), a property studied in some detail by Acciaio and Svindland \cite{acciaio-svindland-concave} which holds for every optimized certainty equivalent. Second, as an interesting decision-theoretic consequence of the defining property (4) of divergences, note that if $T : E \rightarrow F$ is measurable then $\alpha(\nu \circ T^{-1} | \mu \circ T^{-1}) \le \alpha(\nu | \mu)$; we show that equality holds if $T$ is a sufficient statistic for $\{\mu,\nu\}$. Lastly, we show that simplified divergences can be approximated in a sense by their projections on finite sets.

The paper is organized as follows. Section \ref{se:riskmeasures} reviews the basic definitions and duality results of law invariant risk measures before introducing divergences and studying their most essential properties. The main characterization of divergences in terms of law invariant risk measures is given by Proposition \ref{pr:informationinequality} and Theorem \ref{th:divergence-characterization}. Section \ref{se:convexitycontinuitysufficiency} then introduces the concept of a \emph{simplified divergence}, and as an important class of examples we then clarify the connection between continuity properties of a risk measure and joint lower semicontinuity of the induced divergence. This is a useful preparatory step for Section \ref{se:timeconsistency}, which turns to time consistency and superadditivity. The main Theorem \ref{th:mainequivalence} characterizes time consistency properties of a law invariant risk measure in terms of both the additivity properties of the induced divergence and what we call the \emph{shift-convexity} of its measure acceptance set. Section \ref{se:furtherproperties} studies additional results pertaining to convexity and some more information-theoretic uses for divergences, and it should be noted that this section is completely independent of Section \ref{se:timeconsistency}. Finally, Section \ref{se:examples} applies the theory to the examples of shortfall risk measures and optimized certainty equivalents. The short appendix is devoted to the proof of a technical lemma.

\section{Risk measures and divergences} \label{se:riskmeasures}
First, let us fix some notation. Throughout the paper, $(\Omega,\F,P)$ is a fixed probability space, which we assume is a nonatomic standard Borel space. Abbreviate $L^p=L^p(\Omega,\F,P)$ as usual for the set of (equivalence classes of) $p$-integrable real-valued measurable functions on $\Omega$. Let $\P(\Omega)$ denote the set of probability measures on $(\Omega,\F)$, and let $\P_P(\Omega)$ denote the subset consisting of those measures which are absolutely continuous with respect to $P$.
As stated in the introduction, a \emph{risk measure} to us is a convex nondecreasing (with respect to a.s. order) functional $\rho : L^\infty \rightarrow \R$ satisfying $\rho(0)=0$ and $\rho(X+c)=\rho(X)+c$ for all $X\in L^\infty, c \in \R$. Note again that this is somewhat different from the standard definition, in which $\rho$ is instead nonincreasing \cite{follmer-schied-book}.
Law-invariant risk measures possess some nice additional structure, highlighted in particular by the results of \cite{jouini-touzi-schachermayer} and \cite{filipovic-svindland}, though we will not need the latter.

\begin{theorem}[Theorem 2.1 of \cite{jouini-touzi-schachermayer}] \label{th:jouini-touzi-schachermayer}
Every law-invariant risk measure $\rho$ satisfies the \emph{Fatou property}, which means that whenever $X_n \in L^\infty$ are uniformly bounded and converge a.s. to $X \in L^\infty$, then $\rho(X) \le \liminf_{n\rightarrow\infty}\rho(X_n)$.
\end{theorem}

Let us recall a classical duality result, but note that the details of the presentation are somewhat unusual: We say a function $\alpha : \P_P(\Omega) \rightarrow [0,\infty]$ is a \emph{penalty function} for $\rho$ if it holds that
\begin{align}
\rho(X) = \sup_{Q \in \P_P(\Omega)}\left(\E^Q[X] - \alpha(Q)\right).
\label{def:penalty}
\end{align}
(Note that the supremum involves only countably additive measures, and we will make no mention of finite additivity.)
Here $\E^Q$ denotes expectation with respect to the probability $Q$. Expectation under the reference measure $P$ is simply denoted $\E$, and integrals on spaces other than $\Omega$ are written in a more explicit measure-theoretic notation.

\begin{theorem}[Theorem 4.33 of \cite{follmer-schied-book}] \label{th:follmerschied}
Suppose $\rho$ is a law invariant risk measure. Then the function $\alpha : \P_P(\Omega) \rightarrow [0,\infty]$ defined by
\begin{align}
\alpha(Q) := \sup\left\{\E^Q[X] : X \in L^\infty, \ \rho(X) \le 0\right\} = \sup_{X \in L^\infty}\left\{\E^Q[X] - \rho(X)\right\} \label{def:minimalpenalty}
\end{align}
is a penalty function for $\rho$. In fact, it is the \emph{minimal penalty function}, in the sense that any other penalty function $\alpha'$ for $\rho$ satisfies $\alpha \le \alpha'$.
\end{theorem}

Note that the dual representation \eqref{def:minimalpenalty} implies that the minimal penalty function is convex and lower semicontinuous with respect to the total variation topology, as well as the weaker topology $\sigma(\P_P(\Omega),L^\infty)$.\footnote{As usual, when $F$ is a set of real-valued functions on a set $E$, the notation $\sigma(E,F)$ refers to the coarsest topology on $E$ rendering the elements of $F$ continuous.} There is an alternative dual representation more specific to law invariant risk measures, due to Kusuoka \cite{kusuoka2001law} and extended in \cite{frittelli2005law,jouini-touzi-schachermayer}, but we will make no use of this.

\begin{remark} \label{re:equivalenceclasses}
Note that we may afford to be lazy about the fact that $\rho$ is to be evaluated at \emph{equivalence classes}, i.e. elements of $L^\infty$, as opposed to specific measurable functions. For a risk measure $\rho$, we may define $\rho(X) := \rho([X])$ in the obvious way for a measurable function $X : \Omega \rightarrow \R$ by finding the equivalence class $[X] \in L^\infty$ to which $X$ belongs. With this in mind, we may then define $\alpha(Q) := \infty$ for $Q \in \P(\Omega)$ which are not absolutely continuous with respect to $P$, and then the dual formula \eqref{def:penalty} may be re-written
\[
\rho(X) = \sup_{Q \in \P(\Omega)}\left(\E^Q[X] - \alpha(Q)\right),
\]
for bounded measurable functions $X : \Omega \rightarrow \R$.
\end{remark}

As with many properties of convex risk measures, law invariance may be alternatively characterized by a property of the minimal penalty function, and this will be a building block for a more general discussion in the next section. This characterization appears to be new, although a very similar result appeared in \cite[Proposition 2]{shapiro2013kusuoka}, and see also \cite[Lemma A.4]{jouini-touzi-schachermayer}.

\begin{proposition} \label{pr:lawinvariantcharacterization}
Suppose a risk measure $\rho$ on $L^\infty$ has the Fatou property (see Theorem \ref{th:jouini-touzi-schachermayer}). Then $\rho$ is law-invariant if and only if it has a penalty function $\alpha$ satisfying $\alpha(Q \circ T^{-1}) \le \alpha(Q)$ for every $Q \in \P_P(\Omega)$ and for every measurable $T : \Omega \rightarrow \Omega$ satisfying $P \circ T^{-1} = P$.
\end{proposition}
\begin{proof}
First, assume $\rho$ is law invariant, and let $\alpha$ be its minimal penalty function provided by Theorem \ref{th:follmerschied}. Let $T : \Omega \rightarrow \Omega$ be a measurable map satisfying $P \circ T^{-1} = P$. Then $X \circ T$ and $X$ have the same law and thus $\rho(X) = \rho(X \circ T)$ for every $X \in L^\infty$. Hence
\begin{align*}
\alpha(Q \circ T^{-1}) &= \sup_{X \in L^\infty}\left(\E^{Q \circ T^{-1}}[X] - \rho(X)\right) \\
	&= \sup_{X \in L^\infty}\left(\E^{Q}[X \circ T] - \rho(X \circ T)\right) \\
	&\le \sup_{X \in L^\infty}\left(\E^{Q}[X] - \rho(X)\right) \\
	&= \alpha(Q).
\end{align*}
To prove the converse, fix $X,Y \in L^\infty$ with the same law. By \cite[Corollary 6.11]{kallenberg-foundations} (since the probability space is nonatomic) we may find a measurable map $T : \Omega \rightarrow \Omega$ such that $P \circ T^{-1} = P$ and $P(X = Y \circ T)=1$. Then
\begin{align*}
\rho(X) &= \sup_{Q \in \P_P(\Omega)}\left(\E^Q[X] - \alpha(Q)\right) \\
	&= \sup_{Q \in \P_P(\Omega)}\left(\E^{Q \circ T^{-1}}[Y] - \alpha(Q)\right) \\
	&\le \sup_{Q \in \P_P(\Omega)}\left(\E^{Q \circ T^{-1}}[Y] - \alpha(Q \circ T^{-1})\right) \\
	&\le \sup_{Q \in \P_P(\Omega)}\left(\E^Q[Y] - \alpha(Q)\right) \\
	&= \rho(Y).
\end{align*}
Reversing the roles of $X$ and $Y$ completes the proof.
\end{proof}

\begin{remark}
From the proof of Proposition \ref{pr:lawinvariantcharacterization}, it should be clear that the assumption that $\rho$ has the Fatou property is not needed. We state only this simpler form in order to avoid introducing additional terminology, and to avoid dwelling on details involving finitely additive measures.
\end{remark}

\subsection{Divergences and their characterization} \label{se:divergences}

Let us now exploit law invariance to construct a corresponding \emph{family} of risk measures and what we refer to as \emph{divergences}.
Fix for the rest of this section a law-invariant risk measure $\rho$. Given a Polish space $E$, let $\P(E)$ denote the set of Borel probability measures on $E$. 
We write $\nu \ll \mu$ to mean $\nu$ is absolutely continuous with respect to $\mu$. Given any $\mu \in \P(E)$, write $\P_\mu(E) := \{\nu \in \P(E) : \nu \ll \mu\}$. Define also $C(E)$, $C_b(E)$, and $B(E)$ to be the sets of continuous, bounded continuous, and bounded measurable functions on $E$, respectively.
The space $\P(E)$ is endowed with the $\sigma$-field generated by the maps $\mu \mapsto \mu(A)$, where $A \subset E$ is Borel; this equals the Borel $\sigma$-field generated by the topology of weak convergence, i.e., $\sigma(\P(E),C_b(E))$.
%Unless stated otherwise, all topological statements involving $\P(E)$ are with respect to weak convergence, i.e., the topology generated by the maps $\mu \mapsto \int f\,d\mu$ for $f \in C_b(E)$, and all measurability statements with respect to the Borel $\sigma$-field generated by this topology.

Given a Polish space $E$ and $\mu \in \P(E)$, we may find (because $\Omega$ is nonatomic) a measurable function $X : \Omega \rightarrow E$ such that $P \circ X^{-1} = \mu$. We may then define a (law invariant) risk measure $\rho_\mu$ on $L^\infty(E,\mu)$ by
\[
\rho_\mu(f) := \rho(f(X)).
\]
Note that by law-invariance this definition does not depend on the choice of $X$, as long as $P \circ X^{-1} = \mu$. 
We call $(\rho_\mu)_{\mu,E}$ the \emph{family of risk measures induced by $\rho$}.
This family of risk measures satisfies a consistency property, namely 
\begin{align}
\rho_\mu(f) = \rho_\nu(g), \text{ whenever } \mu \circ f^{-1} = \nu \circ g^{-1}. \label{def:riskconsistency}
\end{align}
In particular, for any measurable map $T$ from one Polish space $E$ to another $F$, we have $
\rho_{\mu \circ T^{-1}}(f) = \rho_\mu(f \circ T)$, for $f \in L^\infty(F,\mu \circ T^{-1})$. 
The same construction is valid when $E$ is any standard Borel space, but for simplicity we stick with Polish spaces.

The minimal penalty function of $\rho_\mu$ is denoted $\alpha( \cdot | \mu) : \P_\mu(E) \rightarrow [0,\infty]$ and defined by
\begin{align}
\alpha(\nu | \mu) &:= \sup_{f \in B(E)}\left(\int_E f\,d\nu - \rho_\mu(f)\right) = \sup\left\{\int_E f\,d\nu : f \in L^\infty(E,\mu), \ \rho_\mu(f) \le 0\right\}. \label{def:alpha}
\end{align}
Extend $\alpha(\cdot | \mu)$ to all of $\P(E)$ by setting $\alpha(\nu | \mu) = \infty$ whenever $\nu$ is not absolutely continuous with respect to $\mu$. Then, for $f \in B(E)$,
\[
\rho(f(X)) = \rho_\mu(f) = \sup_{\nu \in \P(E)}\left(\int_E f\,d\nu - \alpha(\nu | \mu)\right), \text{ if } P \circ X^{-1} = \mu,
\]
and it is easy to check that \eqref{def:alpha} remains valid for $\nu \in \P(E) \backslash \P_\mu(E)$.
(As in Remark \ref{re:equivalenceclasses}, let us not be overly careful about distinguishing between measurable functions and equivalence classes thereof.)
We refer to $\alpha(\cdot|\cdot)$ as the \emph{divergence induced by $\rho$}.
Note that $\alpha(\cdot | \cdot)$ is defined for \emph{pairs} of probability measures on \emph{any} Polish space. Additionally, $\alpha(\cdot | \mu)$ is always convex and lower semicontinuous with respect to total variation, and also with respect to the topology $\sigma(\P(E),B(E))$. 
An alternative expression for the divergence induced by $\rho$ is through the \emph{measure acceptance set}
\[
\A := \left\{P \circ X^{-1} : X \in L^\infty, \ \rho(X) \le 0\right\} \subset \P(\R).
\]
Indeed, we may then write
\[
\alpha(\nu|\mu) = \sup\left\{ \int_Ef\,d\nu : f \in B(E), \ \mu \circ f^{-1} \in \A\right\}.
\]

Divergences satisfy a consistency property related to \eqref{def:riskconsistency}, the statement of which requires some notation involving kernels. Given Polish spaces $E$ and $F$, a \emph{kernel from $E$ to $F$} is a measurable function $E \ni x \mapsto K_x \in \P(F)$. Given $\mu \in \P(E)$, write $\mu K := \int_E\mu(dx)K_x(\cdot)$ for the mean measure in $\P(F)$, i.e.,
\[
\mu K(A) = \int_E\mu(dx)K_x(A), \text{ for } A \subset E \text{ Borel}.
\]
For $f \in B(F)$, write $Kf$ for the function $Kf(x) = \int_FK_x(dy)f(y)$ in $B(E)$. Note the identity $\int_F f\,d(\mu K) = \int_E Kf\,d\mu$.

\begin{proposition} \label{pr:informationinequality}
Let $\rho$ be a law invariant risk measure with divergence $\alpha$. If $E$ and $F$ are Polish spaces and $K$ is a kernel from $E$ to $F$, then
\begin{align}
\alpha(\nu K | \mu K) \le \alpha(\nu | \mu), \label{def:informationinequality}
\end{align}
for all $\mu,\nu \in \P(E)$.
In particular, if $T : E \rightarrow F$ is measurable, then
\[
\alpha(\nu \circ T^{-1} | \mu \circ T^{-1}) \le \alpha(\nu | \mu),
\]
and equality holds if $T$ is bijective with measurable inverse.
\end{proposition}
\begin{proof}
Note that the second claim follows from the first by setting $K(x,dy) = \delta_{T(x)}(dy)$. Jensen's inequality shows easily that $\mu \circ (Kf)^{-1} \le (\mu K) \circ f^{-1}$ in convex order for all $f \in B(F)$; indeed, for every convex function $\phi$ on $\R$,
\[
\int_\R\phi\,d\mu \circ (Kf)^{-1} = \int_E\phi(Kf)\,d\mu \le \int_EK(\phi \circ f)\,d\mu = \int_F \phi \circ f\,d(\mu K) = \int_\E\phi\,d(\mu K) \circ f^{-1}.
\]
It is well known that (normalized) law invariant risk measures are increasing with respect to convex order, e.g. by \cite[Corollary 4.65]{follmer-schied-book}, and thus $\rho_{\mu K}(f) \ge \rho_{\mu}(Kf)$. Then
\begin{align*}
\alpha(\nu K | \mu K) &= \sup_{f \in B(F)}\left(\int_Ff\,d(\nu K) - \rho_{\mu K}(f)\right) \\
	&= \sup_{f \in B(F)}\left(\int_E (Kf)\,d\mu - \rho_{\mu K}(f)\right) \\
	&\le \sup_{g \in B(E)}\left(\int_E g\,d\mu - \rho_{\mu}(g)\right) \\
	&= \alpha(\nu | \mu).
\end{align*}
\end{proof}

In fact, the inequality \eqref{def:informationinequality} is enough to reconstruct from $\alpha$ the original family of risk measures. This is made precise in the following:

\begin{theorem} \label{th:divergence-characterization}
Suppose we are given family of functions $\P(E) \ni \nu \mapsto \alpha(\nu|\mu) \in [0,\infty]$, for each Polish space $E$ and each $\mu \in \P(E)$, and suppose the following conditions hold:
\begin{enumerate}
\item $\alpha(\mu|\mu)=0$.
\item $\alpha(\nu|\mu)=\infty$ if $\nu \in \P(E)$ is not absolutely continuous with respect to $\mu$.
\item $\alpha(\nu K | \mu K) \le \alpha(\nu | \mu)$ for every $\nu \in \P(E)$ and every kernel $K$ from $E$ to another Polish space $F$.
\end{enumerate}
For each Polish space $E$ and each $\mu \in \P(E)$, define
\begin{align}
\rho_\mu(f) := \sup_{\nu \in \P(E)}\left(\int_E f\,d\nu - \alpha(\nu | \mu)\right), \ f \in B(E). \label{def:inducedriskmeasurefamily}
\end{align}
Then each $\rho_\mu$ is a law invariant risk measure. Moreover, for any Polish spaces $F$ and $G$, any $\mu \in \P(F)$ and $\nu \in \P(G)$, and any $f \in B(F)$ and $g \in B(G)$ with $\mu \circ f^{-1} = \nu \circ g^{-1}$, we have $\rho_\mu(f) = \rho_\nu(g)$.
\end{theorem}
\begin{proof}
It is immediate from the definition that $\rho_\mu$ is a risk measure. Indeed, since $\alpha(\mu|\mu) = 0$ and $\alpha(\nu|\mu) \ge 0$ for all $\nu$, we have $\rho_\mu(0)=0$. Theorem 4.33 of \cite{follmer-schied-book} shows that $\rho_\mu$ satisfies the Fatou property, since the supremum in its definition includes only countably additive measures. For a fixed $\mu$, we deduce from property (3) and Proposition \ref{pr:lawinvariantcharacterization} that $\rho_\mu$ is law-invariant.

It remains to prove the last claim. 
Suppose for the moment that we can find a kernel $K$ from $F$ to $G$ such that $\mu K = \nu$ and $\mu(Kg = f) =1$. Then
\begin{align*}
\rho_\nu(g) &= \sup_{\eta \in \P(G)}\left(\int_G g\,d\eta - \alpha(\eta | \nu)\right) \\
	&\ge \sup_{\eta \in \P_\mu(F)}\left(\int_G g \,d(\eta K) - \alpha(\eta K | \nu)\right) \\
	&= \sup_{\eta \in \P_\mu(F)}\left(\int_F f \,d\eta - \alpha(\eta K | \mu K)\right) \\
	&\ge \sup_{\eta \in \P_\mu(F)}\left(\int_F f \,d\eta - \alpha(\eta  | \mu )\right) \\
	&= \rho_\mu(f).
\end{align*}
Indeed, the second inequality follows from the assumption (3). Reversing the roles of $f$ and $g$ completes the proof. To prove the existence of such a kernel, we appeal to a famous theorem of Strassen \cite[Theorem 3]{strassen-marginals}: Define $h_\phi$ on $F$ by
\[
h_\phi(x) = \sup_{\eta \in S(x)}\int_G\phi\,d\eta, \quad \text{ where } \quad S(x) := \left\{\eta \in \P(G) : \int_Gg\,d\eta = f(x)\right\}.
\]
If $S(x)$ is nonempty for each $x \in F$, and if $h_\phi$ is measurable, then Strassen's theorem says that there exists a kernel $K$ from $F$ to $G$ satisfying both $\mu K=\nu$ and $\mu(Kg=f)=1$ if and only if
\begin{align}
\int_G\phi\,d\nu \le \int_Fh_\phi\,d\mu, \text{ for all } \phi \in C_b(G). \label{pf:divergence-characterization1}
\end{align}
Suppose for the moment that $S(x)$ is nonempty for each $x \in F$ and that $h_\phi$ is measurable, so that we can apply this theorem. Define a new function $\widetilde{h}_\phi$ on $\R$ by
\begin{align*}
\widetilde{h}_\phi(a) := \sup\phi(g^{-1}(\{a\})) = \sup\left\{\phi(y) : y \in G, \ g(y) = a\right\},
\end{align*}
with the usual convention $\sup\emptyset = -\infty$.
Let us check that
\begin{align}
\widetilde{h}_\phi(f(x)) \le h_\phi(x), \ \mu-a.e. \ x \in F. \label{pf:divergence-characterization2}
\end{align}
If $x \in F$ has $\widetilde{h}_\phi(f(x)) = -\infty$, there is nothing to prove. 
So suppose $x \in F$  satisfies $\widetilde{h}_\phi(f(x)) > -\infty$. For fixed $\epsilon > 0$ we may then choose $y_0 \in G$ such that $g(y_0)=f(x)$ and $\phi(y_0) \ge \widetilde{h}_\phi(f(x)) - \epsilon$. Then, since $\int_G g\,d\delta_{y_0} = f(x)$, we have by definition $\phi(y_0) \le h_\phi(x)$, and thus $\widetilde{h}_\phi(f(x)) \le \epsilon + h_\phi(x)$. Since $\epsilon$ was arbitrary, this proves \eqref{pf:divergence-characterization2}.
Finally, since clearly $\phi(x) \le \widetilde{h}_\phi(g(x))$ for all $x \in G$, using $\nu \circ g^{-1} = \mu \circ f^{-1}$ we prove \eqref{pf:divergence-characterization1}:
\[
\int_G\phi\,d\nu \le \int_G\widetilde{h}_\phi\circ g\,d\nu = \int_F\widetilde{h}_\phi \circ f\,d\mu \le \int_Fh_\phi\,d\mu.
\]
It remains to check the technical points left out above. First note that $S(x)$ is nonempty for $\mu$-almost every $x$: If $S(x)$ is empty then $f(x)$ is not in the range of $g$, which cannot hold on a set of positive $\mu$-measure because $\mu \circ f^{-1} = \nu \circ g^{-1}$. Modify $S(x)$ to equal $F$ on a null set in such a way that it is never empty $S(x) \ne \emptyset$. Next note that $h_\phi$ is universally measurable because the graph of $S(x)$ is analytic \cite[Proposition 7.47]{bertsekasshreve}, so we may apply Strassen's theorem by simply replacing the Borel $\sigma$-field on $F$ with its universal completion.
\end{proof}

With the previous result in mind, it is natural to make the following definition:

\begin{definition}
A \emph{divergence} is a family of convex lower-semicontinuous (with respect to total variation) functions $\P(E) \ni \nu \mapsto \alpha(\nu | \mu) \in [0,\infty]$, for each Polish space $E$ and each $\mu \in \P(E)$, satisfying properties (1-3) of Theorem \ref{th:divergence-characterization}. Given a divergence $\alpha$, the \emph{corresponding (or induced) family of risk measures} is the family $(\rho_\mu)_{\mu,E}$ defined by \eqref{def:inducedriskmeasurefamily}. The \emph{corresponding (or induced) risk measure} is the risk measure $\bar{\rho}$ defined on $L^\infty=L^\infty(\Omega,\F,P)$ by
\[
\bar{\rho}(X) := \rho_{P \circ X^{-1}}(id),
\]
where $id$ denotes the identity map on $\R$. Thanks to Theorem \ref{th:divergence-characterization}, $\bar{\rho}$ is well defined. It is straightforward to check that its induced divergence is exactly $\alpha$, and also that $\bar{\rho}_\mu = \rho_\mu$ for each Polish space $E$ and $\mu \in \P(E)$.
\end{definition}

\subsection{Simplified divergences} \label{se:convexitycontinuitysufficiency}

An important property of relative entropy is that its dual formula can be reduced to a supremum over \emph{continuous} functions: For a Polish space $E$ and for $\mu,\nu \in \P(E)$,
\[
H(\nu | \mu) = \sup_{f \in B(E)}\left(\int_E f\,d\nu - \log\int_Ee^f\,d\mu\right) = \sup_{f \in C_b(E)}\left(\int_E f\,d\nu - \log\int_Ee^f\,d\mu\right).
\]
For our characterization of superadditivity of divergences in Section \ref{se:timeconsistency}, it will be important for us to have a similar result for general divergences. Such a simplification is not always possible, so we make a definition:

\begin{definition} \label{def:simplified}
The divergence $\alpha$ induced by a law invariant risk measure $\rho$ is said to be \emph{simplified} if for every $\mu,\nu \in \P([0,1])$ we have
\[
\alpha(\nu | \mu) = \sup_{f \in C([0,1])}\left(\int_{[0,1]} f\,d\nu - \rho_\mu(f)\right).
\]
Equivalently, for each $\mu \in \P([0,1])$, the map $\alpha(\cdot | \mu)$ is weakly lower semicontinuous, where ``weakly'' refers to the usual weak convergence topology $\sigma(\P(E),C_b(E))$.\footnote{This equivalence is a simple application of the Fenchel-Moreau theorem: Let $M([0,1])$ denote the set of bounded finitely additive signed measures on $[0,1]$, and extend $\alpha(\cdot | \mu)$ to $M([0,1])$ by setting $\bar{\alpha}_\mu(\nu) = \alpha(\nu | \mu)$ for $\nu \in \P([0,1])$ and $\bar{\alpha}_\mu(\nu) = \infty$ otherwise. Put $M([0,1])$ and $C([0,1])$ in duality. Then $\rho_\mu$ is the convex conjugate of $\bar{\alpha}_\mu$. As $\bar{\alpha}_\mu$ is convex and proper, it is lower semicontinuous if and only if it equals its biconjugate.}
\end{definition}

An important reason for this definition is the following measurability result, which we are unfortunately unable to prove without the additional assumption.

\begin{lemma} \label{le:jointlymeasurable}
Every simplified divergence $\alpha$ is jointly measurable, in the sense that for any fixed Polish space $E$ the function $\alpha(\cdot | \cdot)$ is jointly measurable on $\P(E) \times \P(E)$ (with respect to the Borel $\sigma$-field generated by the topology of weak convergence).
\end{lemma}
\begin{proof}
Fix a Polish space $E$. By Borel isomorphism (see \cite[Theorem 15.6]{kechris-settheory}), there exists a measurable bijection $T : E \rightarrow [0,1]$ with measurable inverse. It follows from Proposition \ref{pr:informationinequality} that
\[
\alpha(\nu \circ T^{-1} | \mu \circ T^{-1}) = \alpha(\nu | \mu),
\]
for all $\mu,\nu \in \P(E)$. Note also that the map $\mu \mapsto \mu \circ T^{-1}$ is a measurable bijection from $\P(E)$ to $\P([0,1])$ with measurable inverse. Thus, to show $\alpha$ is jointly measurable on $\P(E) \times \P(E)$, it suffices 
to show it is jointly measurable on $\P([0,1]) \times \P([0,1])$.
Since $\alpha$ is simplified, we have
\[
\alpha(\nu | \mu) = \sup_{f \in C([0,1])}\left(\int f\,d\nu - \rho_\mu(f)\right)
\]
for $\mu,\nu \in \P([0,1])$. Since $\rho_\mu$ is Lipschitz with respect to the supremum norm on $C([0,1])$, and since $C([0,1])$ is separable, we may reduce the supremum above to a countable one. But $\nu \mapsto \int f\,d\nu$ is measurable for each $f \in C([0,1])$, as is $\mu \mapsto \rho_\mu(f)$ thanks to Lemma \ref{le:rhofamilyregularity}.
\end{proof}

\subsection{Semicontinuity of divergences}

The rest of the section studies lower semicontinuity properties of $\alpha$, in part for their intrinsic interest, and in part for a tractable condition that will allow us to verify that all of the examples of divergences we discuss in Section \ref{se:examples} are indeed simplified.
We did not find a good characterization of simplified divergences on the dual side, i.e., in terms of $\rho$, but the following partial results shed some light on the condition nonetheless.
We know that for any divergence $\alpha$, the map $\alpha(\cdot | \mu)$ is lower semicontinuous with respect to the topology $\sigma(\P(E),B(E))$ for any fixed $\mu$, $E$. We will see later in Lemma \ref{pr:strongerlsc} that in fact $\alpha(\cdot|\cdot)$ is \emph{jointly} lower semicontinuous with respect to the same topology. On the other hand, relative entropy is known to be jointly lower semicontinuous with respect to weak convergence, and we first characterize those divergences which share this property. For this it helps to make two definitions, the second of which is well known:

\begin{definition} \label{def:lsc}
We say a divergence $\alpha$ is \emph{jointly weakly lower semicontinuous} if, for each Polish space $E$, the map $\alpha(\cdot|\cdot)$ is lower semicontinuous on $\P(E) \times \P(E)$ with respect to the topology of weak convergence, i.e., equipping $\P(E)$ with the topology $\sigma(\P(E),C_b(E))$.
\end{definition}

\begin{definition} \label{def:lebesgue-continuous}
We say a risk measure $\rho$ is \emph{Lebesgue continuous} if whenever $X_n \in L^\infty$ is a uniformly bound sequence with $X_n \rightarrow X$ a.s. we have $\rho(X_n) \rightarrow \rho(X)$. This is equivalent to the seemingly weaker condition that whenever $X_n,X \in L^\infty$ with $X_n \downarrow X$ a.s. we have $\rho(X_n) \downarrow \rho(X)$ (c.f. Remark 4.25 and Exercise 4.2.2 of \cite{follmer-schied-book}).
\end{definition}

The main result of this section is the following, and the proof is preceded by a preparatory lemma:

\begin{theorem} \label{th:tight-lowersemicontinuity}
Let $\rho$ be a law invariant risk measure with induced divergence $\alpha$. The following are equivalent:
\begin{enumerate}
\item $\alpha$ is jointly weakly lower semicontinuous.
\item For each Polish space $E$ and each $f \in C_b(E)$, the map $\mu \mapsto \rho_\mu(f)$ is continuous.
\item For each Polish space $E$, each $\mu \in \P(E)$, and each $f,f_n \in C_b(E)$ with $f_n \rightarrow f$ pointwise and with $f_n$ uniformly bounded, we have $\rho_\mu(f_n) \rightarrow \rho_\mu(f)$.
\item $\rho$ is Lebesgue continuous.
\end{enumerate}
If these conditions hold, then:
\begin{enumerate}
\item[(5)] For each Polish space $E$ and each $\nu,\mu \in \P(E)$ we have
\[
\alpha(\nu | \mu) = \sup_{f \in C_b(E)}\left(\int_E f\,d\nu - \rho_\mu(f)\right).
\]
\end{enumerate}
\end{theorem}

\begin{lemma} \label{le:rhofamilyregularity}
Let $\rho$ be a law invariant risk measure. Fix a Polish space $E$ and a function $f \in B(E)$. The map $\Phi : \P(E) \rightarrow \R$ given by $\Phi(\mu) := \rho_\mu(f)$ is measurable. Moreover, if $f \in C_b(E)$, then $\Phi$ is lower semicontinuous.
\end{lemma}
\begin{proof}
First we prove the second claim. Let $\mu_n \rightarrow \mu$ in $\P(E)$. By Skorohod representation, we may find $E$-valued random variables $X,X_n$ defined on $\Omega$ with $P \circ X^{-1}=\mu$, $P \circ X_n^{-1} = \mu_n$, and $X_n \rightarrow X$ almost surely. Then $f(X_n) \rightarrow f(X)$ almost surely since $f$ is continuous, and the sequence $f(X_n)$ is uniformly bounded. Thus, the Fatou property (Theorem \ref{th:jouini-touzi-schachermayer}) implies
\[
\rho_\mu(f) = \rho(f(X)) \le \liminf_{n\rightarrow\infty}\rho(f(X_n)) = \liminf_{n\rightarrow\infty}\rho_{\mu_n}(f).
\]
To prove the first claim, find $M > 0$ such that $|f| \le M$, and write $\rho_\mu(f) = \rho_{\mu \circ f^{-1}}(id)$, where $id$ denotes the identity map on $[-M,M]$. According to the previous argument, $m \mapsto \rho_m(id)$ is lower semicontinuous and thus measurable on $\P([-M,M])$. Since also $\mu \mapsto \mu \circ f^{-1}$ is Borel measurable from $\P(E)$ to $\P([-M,M])$ (easily proven using, e.g., \cite[Proposition 7.25]{bertsekasshreve}), we see that $\Phi$ is the composition of two measurable maps.
\end{proof}

\begin{proof}[Proof of Theorem \ref{th:tight-lowersemicontinuity}]
($1 \Rightarrow 2$) Suppose first that $E$ is compact. Let $\mu_n \rightarrow \mu$ in $\P(E)$. We know from Lemma \ref{le:rhofamilyregularity} that $\rho_\mu(f) \le \liminf_{n\rightarrow\infty}\rho_{\mu_n}(f)$, so we show upper semicontinuity. Let $\epsilon > 0$, and find for each $n$ some $\nu_n \in \P(E)$ satisfying 
\[
\rho_{\mu_n}(f) \le \epsilon + \int f\,d\nu_n - \alpha(\nu_n | \mu_n).
\]
Since $E$ is compact, every subsequence admits a further subsequence $\{n_k\}$ such that $\nu_{n_k} \rightarrow \nu$ for some $\nu \in \P(E)$, and lower semicontinuity of $\alpha$ implies
\[
\limsup_{k\rightarrow\infty}\rho_{\mu_{n_k}}(f) \le \epsilon + \int f\,d\nu - \alpha(\nu | \mu) \le \epsilon + \rho_\mu(f).
\]
This shows $\limsup_{n\rightarrow\infty}\rho_{\mu_n}(f) \le \rho_\mu(f)$.
Finally, if $E$ is not necessarily compact, find $M > 0$ such that $\mu_n(|f| \le M) = 1$. Since $[-M,M]$ is compact, the previous result shows that
\[
\rho_{\mu_n}(f) = \rho_{\mu_n \circ f^{-1}}(id) \rightarrow \rho_{\mu \circ f^{-1}}(id) = \rho_\mu(f),
\]
where $id$ is the identity map on $[-M,M]$.

($2 \Rightarrow 3$) Let $f_n,f \in C_b(E)$ be uniformly bounded with $f_n \rightarrow f$ $\mu$-a.s. Then there exists $M > 0$ such that $\mu(|f_n| \le M) = 1$ for all $n$, and so
\[
\rho_\mu(f_n) = \rho_{\mu \circ f_n^{-1}}(id) \rightarrow \rho_{\mu \circ f^{-1}}(id) = \rho_\mu(f),
\]
where $id$ denotes the identity map on $[-M,M]$.

($3 \Rightarrow 4$) Let $X,X_n \in L^\infty$ be uniformly bounded with $X_n \rightarrow X$ a.s. Find $M > 0$ such that $|X_n|\le M$ a.s. for all $n$. Let $E$ denote the (complete separable) metric space of convergent sequences with values in $[-M,M]$ endowed with the supremum metric, and define $\mu$ on $E$ by 
\[
\mu = \prod_{n=1}^\infty P \circ X_n^{-1}.
\]
Let $f_n(x)=x_n$ denote the coordinate maps, and let $f(x) = \lim_{n\rightarrow\infty}x_n$,
for $x=(x_1,x_2,\ldots) \in E$. Then $f$ and $f_n$ are uniformly bounded and continuous, with $f_n \rightarrow f$ pointwise by construction. Since $\mu \circ f_n^{-1} = P \circ X_n^{-1}$ and $\mu \circ f^{-1} = P \circ X^{-1}$, we have
\[
\rho(X_n) = \rho_\mu(f_n) \rightarrow \rho_\mu(f) = \rho(X).
\]

($4 \Rightarrow 5$)  Clearly we have
\[
\alpha(\nu  | \mu ) \ge \sup_{f \in C_b(E)}\left(\int f\,d\nu - \rho_\mu(f)\right).
\]
To show the reverse inequality, fix $\epsilon > 0$ and find $f \in B(E)$ such that
\[
\alpha(\nu  | \mu ) \le \epsilon + \int f\,d\nu - \rho_\mu(f).
\]
Find a bounded sequence $f_n$ of continuous functions with $f_n \rightarrow f$ a.s. Then, using (3) and the bounded convergence theorem, we get
\[
\alpha(\nu  | \mu ) \le \epsilon + \lim_{n\rightarrow\infty}\left(\int f_n\,d\nu - \rho_\mu(f_n)\right) \le \epsilon + \sup_{f \in C_b(E)}\left(\int f\,d\nu - \rho_\mu(f)\right).
\]

($4 \Rightarrow 3$) Obvious.

($3 \Rightarrow 2$)  Let $\mu_n \rightarrow \mu$ in $\P(E)$ and $f \in C_b(E)$. Let $\lambda$ denote Lebesgue measure on $[0,1]$, and let $q_n$ and $q$ denote the quantile functions corresponding to $\mu_n \circ f^{-1}$ and $\mu\circ f^{-1}$, respectively, so that $\mu_n \circ f^{-1} = \lambda \circ q_n^{-1}$ and $\mu \circ f^{-1} = \lambda \circ q^{-1}$. Then $q_n$ are uniformly essentially bounded with $q_n \rightarrow q$ $\lambda$-a.s., and law invariance yields
\[
\rho_{\mu_n}(f) = \rho_\lambda(q_n) \rightarrow \rho_\lambda(q) = \rho_\mu(f).
\]

($4 \Rightarrow 1$) We know by now that (4) implies both (5) and (2), and thus we can write
\[
\alpha(\nu | \mu) = \sup_{f \in C_b(E)}\left(\int f\,d\nu - \rho_\mu(f)\right).
\]
Since the map $(\mu,\nu) \mapsto \int f\,d\nu - \rho_\mu(f)$ is jointly continuous by (2), $\alpha$ is lower semicontinuous as the supremum of continuous functions.
\end{proof}

As announced before, there are some additional continuity properties of potential interest, although we shall not use these in the sequel.
Note that Lemma \ref{le:jointlymeasurable} does not follow from the following Proposition \ref{pr:strongerlsc}, because the Borel $\sigma$-field of $\sigma(\P(E),B(E))$ is typically strictly larger than the Borel $\sigma$-field of the topology of weak convergence.

\begin{proposition} \label{pr:strongerlsc}
Suppose $\rho$ is a law invariant risk measure with induced divergence $\alpha$.
If $\P(E)$ is endowed with the topology $\sigma(\P(E),B(E))$, then the map $\mu \mapsto \rho_\mu(f)$ is continuous for every $f \in B(E)$, and $\alpha(\cdot|\cdot)$ is lower semicontinuous with respect to the product topology on $\P(E) \times \P(E)$.
\end{proposition}
\begin{proof}
First, fix $f \in B(E)$ with $|f| \le M$ for $M > 0$. Note that $\rho_\mu(f) = \rho_{\mu \circ f^{-1}}(id)$, where $id$ denotes the identity map on $[-M,M]$. We saw in Lemma \ref{le:rhofamilyregularity} that $\P([-M,M]) \ni m \mapsto \rho_m(f)$ is continuous with respect to weak convergence. It is easy to check that $\mu \mapsto \mu \circ f^{-1}$ is a continuous map from $(\P(E),\sigma(\P(E),B(E)))$ to $\P([-M,M])$ endowed with the weak convergence topology, and this proves the first claim. According to the definition \eqref{def:alpha}, $\alpha(\nu |\mu)$ is the supremum of continuous functions of $(\nu,\mu)$, and this proves the second claim.
\end{proof}

\section{Acceptance consistency and superadditivity} \label{se:timeconsistency}

As was first observed by Weber \cite{weber-distributioninvariant}, a law invariant risk measure naturally gives rise to a \emph{dynamic risk measure} on any (nice enough) filtered probability space. We will use the same construction:
Define $\tilde{\rho}$ again by $\tilde{\rho}(P \circ X^{-1}) = \rho(X)$, which makes sense thanks to law-invariance. Using our previous notation, note that $\tilde{\rho}(m) = \rho_m(id)$, where $id$ denotes the identity map on $\R$. We may then define, for any $\sigma$-field $\G \subset \F$ in $\Omega$ and any $X \in L^\infty$,
\[
\rho(X | \G) := \tilde{\rho}(P(X \in \cdot \, | \, \G)).
\]
Note that a regular conditional law of $X$ given $\G$ exists because $\Omega$ is standard. Lemma \ref{le:rhofamilyregularity} ensures that $\rho(X | \G)$ is a $\G$-measurable random variable, defined uniquely up to a.s. equality. Similarly, for a random variable $Y$, write $\rho(X | Y) := \rho(X | \sigma(Y))$. Note that if $Y$ is $\G$-measurable then
\[
\rho(X +Y | \G) = \rho(X | \G) + Y, \ a.s.,
\]
for any random variable $X$. If $X$ and $Y$ are independent, then it is straightforward to check that
\[
\rho(f(X,Y) | X) = \rho(f(x,Y))|_{x = X}.
\]
We are nearly ready to define the type of time-consistency we investigate.

\begin{definition}
We say that a law-invariant risk measure $\rho$ is \emph{acceptance consistent} if
\[
\rho(X) \le \rho(\rho(X | \G)),
\]
for all sub-$\sigma$-fields $\G \subset \F$ and all $X \in L^\infty$. If the inequality is reversed, we say $\rho$ is \emph{rejection consistent}. We say $\rho$ is \emph{time consistent} if it is both acceptance and rejection consistent.
\end{definition}

\begin{remark}
This definition begins to look more like the one appearing in the literature (see \cite{acciaio-penner-dynamic}) once it is applied inductively. Let $(\F_t)_{t \ge 0}$ denote any filtration on $\Omega$, with $\F_t \subset \F$ for all $t$. Indeed, $(\rho(\cdot | \F_t))_{t \ge 0}$ is a \emph{dynamic risk measure} in the sense of \cite{acciaio-penner-dynamic}. If $\rho$ is acceptance consistent and $X \in L^\infty$, then it is straightforward to check that $\rho(X | \F_s) \le \rho(\rho(X | \F_t) | \F_s)$ a.s. for $0 \le s \le t$.
\end{remark}

\subsection{Superadditivity and shift-convexity}
Let us give names to certain divergence inequalities resembling the chain rule of classical relative entropy.
Henceforth we will need to assume that our divergences are simplified, as in Definition \ref{def:simplified}.
As far as the following definition of superadditivity is concerned,
this assumption is merely to ensure that the divergence $\alpha(\cdot | \cdot)$ is jointly measurable, so that the integrals make sense. Later, a technical point in the proof of the main Theorem \ref{th:mainequivalence} will depend crucially on the divergence being simplified, but the question of whether or not Theorem \ref{th:mainequivalence} holds in more generality remains open.

\begin{definition}
We say that a divergence $\alpha$ is \emph{partially superadditive} (resp. \emph{partially subadditive}) if
\[
\alpha\left(\left.\nu(dx)K^\nu_x(dy) \right| \mu_1 \times \mu_2\right) \ge  \alpha(\nu | \mu_1) + \int\nu(dx)\alpha(K^\nu_x | \mu_2), \ \text{ (resp. } \le \text{) }
\]
whenever $\nu(dx)K^\nu_x(dy)$ and $\mu_1 \times \mu_2$ are probability measures on the product of two Polish spaces; note that the latter is required to be a product measure. 
We say a simplified divergence $\alpha$ is (fully) \emph{superadditive} (resp. \emph{subadditive}) if
\[
\alpha\left(\left. \nu(dx)K^\nu_x(dy) \right|  \mu(dx)K^\mu_x(dy)\right) \ge  \alpha(\nu | \mu) + \int\nu(dx)\alpha(K^\nu_x | K^\mu_x), \ \text{ (resp. } \le \text{) }
\]
whenever $\nu(dx)K^\nu_x(dy)$ and $\mu(dx)K^\mu_x(dy)$ are probability measures on the product of two Polish spaces.
\end{definition}

Note that partial superadditivity as opposed to full superadditivity only requires the inequality to hold when the reference measure is a product.
It turns out that these conditions are equivalent, although we have only an indirect proof of this fact. As was discussed in the introduction, additivity properties of a divergence $\alpha$ are linked with time consistency and sub-level set properties of its induced risk measure, which we now describe. In the following, we will write $\P[-M,M]$ for $M > 0$ for to the set of probability measures on $\R$ which are supported on the interval $[-M,M]$.

\begin{definition} {\ }
\begin{enumerate}
\item The \emph{measure acceptance set} $\A$ of a law invariant risk measure $\rho$ is defined by $\A := \{P \circ X^{-1} : X \in L^\infty, \ \rho(X) \le 0\}$. In words, this is the set of laws of random variables $X$ satisfying $\rho(X) \le 0$.
\item A set $\A \subset \P(\R)$ is \emph{shift-convex} if for every $\mu \in \A$, every $M > 0$, and every measurable map $\R \ni x \mapsto K_x \in \A \cap \P[-M,M]$, it holds that the measure $\int_\R\mu(dx)K_x(\cdot - x)$ is in $\A$.
\end{enumerate}
\end{definition}

As was discussed by Weber, the \emph{convexity} of a measure acceptance set $\A$ admits a natural interpretation in terms of so-called compound lotteries: If two outcomes $X$ and $Y$ are acceptable, then convexity of $\A$ means that the outcome with law $tP\circ X^{-1} + (1-t)P\circ Y^{-1}$ is also acceptable, for any $t \in (0,1)$. Shift-convexity is open to interpretation on similar grounds:
Suppose $X$ is an acceptable outcome, and that $Y$ is conditionally acceptable given $X$. Then shift-convexity means that $X+Y$ is itself acceptable. To see this, in the definition of shift-convexity take $\mu$ to be the law of $X$ and $K_x$ to be the conditional law of $Y$ given $X$.
Section \ref{se:shiftconvexity} we will elaborate on interpretations and reformulations of this unusual property. We can now state the main result of this section.

\begin{theorem} \label{th:mainequivalence}
Suppose $\alpha$ is a simplified divergence induced by a law invariant risk measure $\rho$ with acceptance set $\A$. The following are equivalent:
\begin{enumerate}
\item $\rho$ is acceptance consistent.
\item $\A$ is shift-convex.
\item $\alpha$ is superadditive.
\item $\alpha$ is partially superadditive.
\end{enumerate}
Similarly, the same equivalences hold when ``acceptance'' is changed to ``rejection'', ``superadditive'' is changed to ``subadditive'', and $\A$ is changed to $\A^c$. The equivalence of (1) and (2) holds without the assumption that $\alpha$ is simplified.
\end{theorem}

The equivalence between (1) and (3) is related to Theorem 27 of \cite{acciaio-penner-dynamic}, and Section \ref{se:essentialsuprema} elaborates on the precise connection.
From Theorem \ref{th:mainequivalence} we can conclude that \emph{not many divergences are additive}, i.e., both superadditive and subadditive.
According to Kupper and Schachermayer \cite{kupper-schachermayer},  the only time consistent law invariant risk measures are entropic:

\begin{corollary} \label{co:entropic}
Suppose $\alpha$ is a simplified divergence. If $\alpha$ is both superadditive and subadditive, then it is of one of the following forms:
\begin{align*}
\alpha(\nu | \mu) &= \frac{1}{\eta}H(\nu | \mu), \text{ for some } \eta > 0, \ \ \text{ for } \nu \ll \mu, \quad \infty \text{ otherwise}, \\
\alpha(\nu | \mu) &= 0  \ \ \text{ for } \nu = \mu, \quad \infty \text{ otherwise}, \\
\alpha(\nu | \mu) &= 0  \ \ \text{ for } \nu \ll \mu, \quad \infty \text{ otherwise}.
\end{align*}
The induced risk measure in these cases are $\rho(X) = \eta^{-1}\log\E[e^{\eta X}]$, $\rho(X) = \E X$, and $\rho(X) = \esssup X$.
\end{corollary}

\subsection{Properties of time consistency}

The following lemma shows that acceptance consistency is equivalent to a seemingly weaker statement, which will be easier to connect with shift-convexity:

\begin{lemma} \label{le:partialacceptanceconsistency}
Let $\rho$ be a law invariant risk measure. Then $\rho$ is acceptance consistent if and only if the following holds: For every pair of independent random variables $X,Y$ with values in some Polish spaces $E,F$, and for every $f \in B(E \times F)$, we have
\[
\rho(f(X,Y)) \le \rho(\rho(f(X,Y) | X)).
\]
\end{lemma}
\begin{proof}
The ``only if'' direction is immediate. To prove the converse, fix $X \in L^\infty$ and a $\sigma$-field $\G \subset \F$. Find $Y \in L^\infty$ which generates $\G$, for example $Y = \sum_{n=1}^\infty 2^{-n}1_{B_n}$ where $\{B_n\}$ is a countable family of generators of $\G$ (recall that our ambient probability space is standard). By \cite[Theorem 5.10]{kallenberg-foundations}, we may find independent random variables $\widetilde{Y}$ and $U$ as well as a measurable function $f$ such that $(\widetilde{Y},f(\widetilde{Y},U))$ has the same law as $(Y,Z)$. Then the hypothesis and law invariance imply
\[
\rho(Z) = \rho(f(\widetilde{Y},U)) \le \rho(\rho(f(\widetilde{Y},U) | \widetilde{Y})).
\]
But the conditional law of $f(\widetilde{Y},U)$ given $\widetilde{Y}$ is the same as the conditional law of $Z$ given $Y$, and thus law invariance of $\rho$ implies that $\rho(f(\widetilde{Y},U) | \widetilde{Y})$ and $\rho(Z | Y) = \rho(Z | \G)$ have the same law. Using law invariance once more, we conclude that
\[
\rho(\rho(f(\widetilde{Y},U) | \widetilde{Y})) = \rho(\rho(Z | \G)).
\]
\end{proof}

The next lemma rephrases acceptance consistency, in a more measure-theoretic notation which will be useful later.

\begin{proposition} \label{pr:acceptanceconsistent-equivalences}
For a law invariant risk measure $\rho$, the following are equivalent:
\begin{enumerate}
\item $\rho$ is acceptance consistent.
\item For Polish spaces $E$ and $F$, $\bar{\mu}=\mu(dx)K^\mu_x(dy) \in \P(E \times F)$, $f \in B(E \times F)$, and $g \in B(E)$ satisfying $\rho_\mu(g) \le 0$, we have
\[
\mu\left\{x \in E : \rho_{K^\mu_x}(f(x,\cdot)) \le g(x)\right\} = 1 \quad \Rightarrow \quad \rho_{\bar{\mu}}(f) \le 0.
\]
\item For Polish spaces $E$ and $F$, $\bar{\mu}=\mu(dx)K^\mu_x(dy) \in \P(E \times F)$,  and $f \in B(E \times F)$, we have
\[
\rho_{\mu}\left(\rho_{K^\mu_x}(f(x,\cdot))|_{x = X}\right) \ge \rho_{\bar{\mu}}(f),
\]
\item For Polish spaces $E$ and $F$, $\mu_1 \in \P(E)$,  $\mu_2 \in \P(F)$, $f \in B(E \times F)$, and $g \in B(E)$ satisfying $\rho_{\mu_1}(g) \le 0$, we have
\[
\mu\left\{x \in E : \rho_{\mu_2}(f(x,\cdot)) \le g(x)\right\} = 1 \quad \Rightarrow \quad \rho_{\mu_1 \times \mu_2}(f) \le 0.
\]
\item For Polish spaces $E$ and $F$, $\bar{\mu}=\mu_1 \times \mu_2 \in \P(E \times F)$, and $f \in B(E \times F)$, we have
\[
\rho_{\mu_1}\left(\rho_{\mu_2}(f(x,\cdot))|_{x = X}\right) \ge \rho_{\bar{\mu}}(f),
\]
where $X$ denotes the identity map on $E$.
\end{enumerate}
The same equivalences hold for rejection consistent, but with the inequalities reversed.
\end{proposition}
\begin{proof} It is obvious that (3) implies (5) and (2) implies (4). Property (5) and the property described in Lemma \ref{le:partialacceptanceconsistency} are equivalent, merely written in different notation, and thus (5) and (1) are equivalent. It remains to prove $1 \Rightarrow 2 \Rightarrow 3$ and $4 \Rightarrow 5$.

($1 \Rightarrow 2$) Fix Polish spaces $E$ and $F$, $\bar{\mu}=\mu(dx)K^\mu_x(dy) \in \P(E \times F)$, $f \in B(E \times F)$, and $g \in B(E)$ satisfying $\rho_\mu(g) \le 0$. Suppose also
\[
\mu\left\{x \in E : \rho_{K^\mu_x}(f(x,\cdot)) \le g(x)\right\} = 1.
\]
Find an $E \times F$-valued random variable $(X,Y)$ with law $\bar{\mu}$, and note that
\[
\rho(f(X,Y) | X) = \rho_{K^\mu_X}(f(X,\cdot)) \le g(X), \ a.s.
\]
Acceptance consistency and monotonicity of $\rho$ yield
\[
\rho_{\bar{\mu}}(f) = \rho(f(X,Y)) \le \rho(\rho(f(X,Y) | X)) \le \rho(g(X)) = \rho_\mu(g) \le 0.
\]

($2 \Rightarrow 3$) Fix Polish spaces $E$ and $F$, $\bar{\mu}=\mu(dx)K^\mu_x(dy) \in \P(E \times F)$, and $f \in B(E \times F)$. Define $g \in B(E)$ by $g(x) = \rho_{K^\mu_x}(f(x,\cdot))$. Then trivially
\[
\mu\left\{x \in E : \rho_{K^\mu_x}(f(x,\cdot) - \rho_\mu(g)) \le g(x) - \rho_\mu(g)\right\} = 1.
\]
Since $\rho_\mu(g - \rho_\mu(g)) = 0$,
property (2) implies
\[
\rho_{\bar{\mu}}(f - \rho_\mu(g)) \le 0.
\]
Rearrange this to get $\rho_{\bar{\mu}}(f) \le \rho_\mu(g)$, as desired.

($4 \Rightarrow 5$) Fix Polish spaces $E$ and $F$, $\mu_1 \in \P(E)$, $\mu_2 \in \P(F)$, and $f \in B(E \times F)$. Define $g \in B(E)$ by $g(x) = \rho_{\mu_2}(f(x,\cdot))$. Then trivially
\[
\mu_1\left\{x \in E : \rho_{\mu_2}(f(x,\cdot) - \rho_\mu(g)) \le g(x) - \rho_{\mu_1}(g)\right\} = 1.
\]
Since $\rho_{\mu_1}(g - \rho_{\mu_1}(g)) = 0$,
property (4) implies
\[
\rho_{\mu_1 \times \mu_2}(f - \rho_{\mu_1}(g)) \le 0.
\]
Rearrange this to get $\rho_{\mu_1 \times \mu_2}(f) \le \rho_{\mu_1}(g)$, as desired.
\end{proof}

This alternative description of acceptance consistency will serve us especially well when addressing additivity. For now, we will use it in establishing the connection between acceptance consistency and shift-convexity.

\begin{proposition} \label{pr:timeconsistent-shiftconvex}
A law-invariant risk measure is acceptance consistent if and only if it its measure acceptance set is shift-convex.
\end{proposition}
\begin{proof}
Let $\rho$ be a law-invariant risk measure with measure acceptance set $\A$. First, assume $\rho$ is acceptance consistent. Fix $\mu \in \A$, $M > 0$, and a measurable map $\R \ni x \mapsto K_x \in \A \cap \P[-M,M]$. Set $\bar{\mu} = \mu(dx)K_x(dy - x)$. Letting $\lambda$ denote Lebesgue measure on $[0,1]$, we may find (e.g. by \cite[Theorem 5.10]{kallenberg-foundations}) a measurable function $f : \R \times [0,1] \rightarrow \R$ such that, if $\hat{f}(x,y) := (x,f(x,y))$, then 
\[
\bar{\mu} = (\mu \times \lambda) \circ \hat{f}^{-1}.
\]
Now set $g(x) = x$. Since $\lambda \circ f(x,\cdot)^{-1} = K_x(\cdot - x)$, we have
\[
\lambda \circ [f(x,\cdot)-g(x)]^{-1} = K_x \in \A \cap \P[-M,M],
\]
for each $x$. Thus
\[
\mu\left\{x \in \R : \rho_\lambda(f(x,\cdot)) \le g(x)\right\} = \mu\left\{x \in \R : \lambda \circ [f(x,\cdot)-g(x)]^{-1} \in \A\right\} = 1.
\]
Note that since $\mu$ has compact support and $K_x \in \P[-M,M]$ for all $x$, it follows that $f$ is essentially bounded with respect to $\mu \times \lambda$. 
Since also $\mu \circ g^{-1} = \mu \in \A$, i.e., $\rho_\mu(g) \le 0$, acceptance consistency (Proposition \ref{pr:acceptanceconsistent-equivalences}(2)) implies that $\rho_{\mu \times \lambda}(f) \le 0$. In other words, $(\mu \times \lambda) \circ f^{-1} \in \A$. But this completes the proof of shift-convexity, since
\[
(\mu \times \lambda) \circ f^{-1} = \int_{\R}\mu(dx)K_x(\cdot - x).
\]

Conversely, assume now that $\rho$ is shift-convex. Let $E$ and $F$ be Polish spaces, and fix $\mu_1 \in \P(E)$, $\mu_2 \in \P(F)$, $f \in B(E \times F)$, and $g \in B(E)$ with $\rho_{\mu_1}(g) \le 0$. Suppose also that 
\[
\mu_1\left\{x \in E : \rho_{\mu_2}(f(x,\cdot)) \le g(x)\right\} = 1,
\]
or equivalently that 
\[
\mu_1\left\{x \in E : \mu_2 \circ [f(x,\cdot)-g(x)]^{-1} \in \A\right\} = 1,
\]
In light of Proposition \ref{pr:acceptanceconsistent-equivalences}(4), we must check that $\rho_{\mu_1 \times \mu_2}(f) \le 0$, or equivalently that $(\mu_1 \times \mu_2) \circ f^{-1} \in \A$. Set $\nu := \mu_1 \circ g^{-1}$, and note that $\nu \in \A$. For $x \in \R$, define also
\[
K_x := \begin{cases}
\mu_2 \circ [f(x,\cdot)-g(x)]^{-1} &\text{if } \mu_2 \circ [f(x,\cdot)-g(x)]^{-1} \in \A \\
\delta_0 &\text{otherwise}.
\end{cases}
\]
(The choice of $\delta_0$ is arbitrary, and any other element of $\A$ would do.) Then $K_x \in \A$ for each $x$, and shift-convexity implies
\[
\int_\R \nu(dx)K_x(\cdot - x) \in \A.
\]
But in fact $\int_\R \nu(dx)K_x(\cdot - x)$ is equal to $(\mu_1 \times \mu_2) \circ f^{-1}$, since for $\phi \in B(\R)$ we have
\begin{align*}
\int_\R \nu(dx)\int_\R K_x(dy - x)\phi(y) &= \int_\R \nu(dx)\int_\R K_x(dy)\phi(x+y) \\
	&= \int_E\mu_1(dx)\int_F\mu_2(dy)\phi\left(g(x) + f(x,y) - g(x)\right) \\
	&= \int_{E \times F} \phi \circ f\,d(\mu_1 \times \mu_2).
\end{align*}
\end{proof}

Finally, before we turn to the proof of Theorem \ref{th:mainequivalence}, we compute a penalty function for the risk measure $X \mapsto \rho(\rho(X | \G))$, under no time consistency assumptions. This is related to some results in \cite{acciaio-penner-dynamic} and \cite{cheridito2011composition}, to name but a few, but different in the sense that our conditional penalty functions are defined as pointwise suprema as opposed to essential suprema; see Section \ref{se:essentialsuprema} for a discussion of this point.

\begin{proposition} \label{pr:keyidentity}
Let $\rho$ be a law invariant risk measure with induced divergence $\alpha$, which we assume is simplified. Let $E$ and $F$ be Polish spaces, and let $\bar{\mu} = \mu(dx)K^\mu_x(dy) \in \P(E \times F)$. Let $f \in B(E \times F)$, and let $X$ denote the identity map on $E$. Then
\[
\rho_{\mu}\left(\rho_{K^\mu_x}(f(x,\cdot))|_{x = X}\right) = \sup_{\nu(dx)K^\nu_x(dy) \in \P(E \times F)}\left\{\int_E\int_Ff(x,y)K^\nu_x(dy)\nu(dx) - \int_F\alpha(K^\nu_x | K^\mu_x) \nu(dx) - \alpha(\nu | \mu)\right\}.
\]
\end{proposition}
\begin{proof}
We first compute
\begin{align*}
\rho_{\mu}&\left(\rho_{K^\mu_x}(f(x,\cdot))|_{x = X}\right) \\
	&= \sup_{\nu \in \P(E)}\left\{\int_E\rho_{K^\mu_x}(f(x,\cdot))\nu(dx) - \alpha(\nu | \mu)\right\} \\
	&= \sup_{\nu \in \P(E)}\left\{\int_E\sup_{\eta \in \P(F)}\left(\int_Ff(x,y)\eta(dy) - \alpha(\eta | K^\mu_x) \right)\nu(dx) - \alpha(\nu | \mu)\right\}.
\end{align*}
Complete the proof by using a well known measurable selection argument \cite[Proposition 7.50]{bertsekasshreve} to deduce
\begin{align*}
\int_E&\sup_{\eta \in \P(F)}\left(\int_Ff(x,y)\eta(dy) - \alpha(\eta | K^\mu_x) \right)\nu(dx) \\
	&= \sup_{K}\left(\int_E\int_Ff(x,y)K_x(dy)\nu(dx) - \int_F\alpha(K_x | K^\mu_x) \nu(dx,)\right),
\end{align*}
where the supremum is over all kernels from $E$ to $F$.
\end{proof}

\subsection{Proof of Theorem \ref{th:mainequivalence}}
We saw in Proposition \ref{pr:timeconsistent-shiftconvex} that acceptance consistency and shift-convexity are equivalent. We will prove that acceptance consistency implies superadditivity and that partial superadditivity implies acceptance consistency. This is enough, since clearly superadditivity implies partial superadditivity. Fix throughout two Polish spaces $E$ and $F$ and a function $f \in B(E \times F)$.

First assume $\rho$ is partially superadditive. Fix $\bar{\mu} = \mu_1 \times \mu_2 \in \P(E \times F)$. Use Proposition \ref{pr:keyidentity} followed by partial superadditivity to get
\begin{align*}
\rho_{\mu}&\left(\rho_{\mu_2}(f(x,Y))|_{x = X}\right) \\
	&= \sup_{\nu(dx)K^\nu_x(dy) \in \P(E \times F)}\left\{\int_E\int_Ff(x,y)K^\nu_x(dy)\nu(dx) - \int_F\alpha(K^\nu_x | \mu_2) \nu(dx) - \alpha(\nu | \mu_1)\right\} \\
	&\ge \sup_{\nu(dx)K^\nu_x(dy) \in \P(E \times F)}\left\{\int_E\int_Ff(x,y)K^\nu_x(dy)\nu(dx) -  \alpha(\nu(dx)K^\nu_x(dy) | \bar{\mu})\right\} \\
	&= \rho_{\bar{\mu}}(f).
\end{align*}
Conclude from Proposition \ref{pr:acceptanceconsistent-equivalences}(5) that $\rho$ is acceptance consistent.

Now suppose $\rho$ is acceptance consistent. Let $\bar{\mu} = \mu(dx)K^\mu_x(dy) \in \P(E \times F)$. Use Proposition \ref{pr:keyidentity} followed by Proposition \ref{pr:acceptanceconsistent-equivalences}(3) to get
\begin{align}
\sup_{\nu(dx)K^\nu_x(dy) \in \P(E \times F)}&\left\{\int_E\int_Ff(x,y)K^\nu_x(dy)\nu(dx) - \int_F\alpha(K^\nu_x | K^\mu_x) \nu(dx) - \alpha(\nu | \mu)\right\} \nonumber \\
	&= \rho_{\mu}\left(\rho_{K^\mu_x}(f(x,\cdot))|_{x = X}\right) \label{pf:mainthm1} \\
	&\ge \rho_{\bar{\mu}}(f) \nonumber
\end{align}
On the other hand, according to Lemma \ref{le:integralconvexity} proven below, and using the definition of $\alpha$, we get
\begin{align*}
\int\nu(dx)\alpha(K^\nu_x|K^\mu_x) + \alpha(\nu | \mu) &= \sup_{f \in B(E \times F), g \in B(E)}\left\{\int f\,d\bar{\nu} - \int\nu(dx)\rho_{K^\mu_x}(f(x,\cdot)) + \int g\,d\nu - \rho_\mu(g)\right\} \\
	&= \sup_{f \in B(E \times F), g \in B(E)}\left\{\int_{E \times F} (g(x) + f(x,y))\,\bar{\nu}(dx,dy) - \rho_\mu(g + \rho_{K^\mu_x}(f(x,\cdot)))\right\} \\
	&= \sup_{f \in B(E \times F), g \in B(E)}\left\{\int f\,d\bar{\nu} - \rho_\mu(\rho_{K^\mu_x}(f(x,\cdot))|_{x = X})\right\}
\end{align*}
Indeed, in the second line we replaced $g(x)$ by $g(x) + \rho_{K^\mu_x}(f(x,\cdot))$, and in the final step we replaced $f$ by $f+g$. This shows that the function of $\nu(dx)K^\nu_x(dy)\in\P(E \times F)$ given by
\begin{align}
\int_F\alpha(K^\nu_x | K^\mu_x) \nu(dx) + \alpha(\nu | \mu) \label{pf:mainthm2}
\end{align}
is precisely the minimal penalty function of the risk measure given by \eqref{pf:mainthm1}.
Since $\bar{\nu} \mapsto \alpha(\bar{\nu}|\bar{\mu})$ is the minimal penalty function of $\rho_{\bar{\mu}}$ (see Theorem \ref{th:follmerschied}), it follows from the order-reversing property of convex conjugation that $\alpha(\cdot|\bar{\mu})$ dominates the minimal penalty function of the risk measure given by \eqref{pf:mainthm1}. That is,
\[
\alpha\left(\left. \nu(dx)K^\nu_x(dy) \right| \bar{\mu}\right) \ge \int_F\alpha(K^\nu_x | K^\mu_x) \nu(dx) + \alpha(\nu | \mu),
\]
for every $\nu(dx)K^\nu_x(dy) \in \P(E \times F)$. \hfill \qedsymbol

\begin{lemma}[Nearly Lemma 4 of \cite{acciaio-penner-dynamic}] \label{le:integralconvexity}
For any $\bar{\nu} = \nu(dx)K^\nu_x(dy) \in \P(E \times F)$, we have 
\[
\int\nu(dx)\alpha(K^\nu_x|K^\mu_x) = \sup_{f \in B(E \times F)}\left\{\int f\,d\bar{\nu} - \int\nu(dx)\rho_{K^\mu_x}(f(x,\cdot))\right\}
\]
\end{lemma}

\subsection{Essential suprema} \label{se:essentialsuprema}
Let us briefly discuss how to connect our results with a more common dual characterization of acceptance consistency in terms of penalty functions, as can be found in \cite{acciaio-penner-dynamic}. Assume throughout that our divergence $\alpha$ is simplified.
Let $\bar{\mu} = \mu(dx)K^\mu_x(dy)\in \P(E \times F)$ for Polish spaces $E$ and $F$. Let $(X,Y)$ denote the identity map (i.e., coordinate maps) on $E \times F$, and define the filtration $(\F_0,\F_1,\F_2)$ on $E \times F$ by letting $\F_0$ be the trivial $\sigma$-field, letting $\F_1 = \sigma(X)$, and letting $\F_2$ be the Borel $\sigma$-field.
Define a dynamic risk measure $(\rho_0,\rho_1)$ on $E \times F$ by 
\begin{align*}
[\rho_1(f)](x,y) &= \rho_{K^\mu_x}(f(x,\cdot)), \text{ for } (x,y) \in E \times F \\
\rho_0(f) &= \rho_{\bar{\mu}}(f).
\end{align*}
That is, $\rho_1$ maps $\F_2$-measurable random variables to $\F_1$-measurable random variables. Alternatively, we could see $\rho_1$ as mapping from $L^\infty(E \times F,\bar{\mu})$ to $L^\infty(E,\mu)$.
In this notation, acceptance consistency simply means $\rho_0(f) \le \rho_0(\rho_1(f))$ for all $f \in B(E \times F)$. 
According to Theorem 27 of \cite{acciaio-penner-dynamic}, acceptance consistency is equivalent to the inequality
\[
\alpha_0(\bar{\nu}) \ge \alpha_{0,1}(\bar{\nu}) + \E^{\bar{\nu}}\left[\alpha_{1}(\bar{\nu}) \right]
\]
holding for every $\bar{\nu} = \nu(dx)K^\nu_x(dy) \in \P_{\bar{\mu}}(E \times F)$, where $\alpha_0$, $\alpha_1$, and $\alpha_{0,1}$ are defined by 
\begin{align*}
\alpha_0(\bar{\nu}) &:= \sup\left\{\E^{\bar{\nu}}[f] : f \in L^\infty(E \times F,
\F_2,\bar{\mu}), \ \rho_0(f) \le 0\right\} = \alpha(\bar{\nu} | \bar{\mu}), \\
\alpha_1(\bar{\nu}) &:= \esssup\left\{\E^{\bar{\nu}}[f | \F_1] : f \in L^\infty(E \times F,\F_2,\bar{\mu}), \ \rho_1(f) \le 0 \ a.s.\right\} \\
	&= \esssup\left\{\int_F f(X,y)\,K^\nu_X(dy) : f \in L^\infty(E \times F,\F_2,\bar{\mu}), \ \mu\{ x\in E :\rho_{K^\mu_x}(f(x,\cdot)) \le 0\} = 1\right\} \\
\alpha_{0,1}(\bar{\nu}) &:= \sup\left\{\E^{\bar{\nu}}[f] : f \in L^\infty(E \times F,
\F_1,\bar{\mu}), \ \rho_0(f) \le 0\right\} \\
	&= \sup\left\{\int_E f\,d\nu : f \in L^\infty(E,
\mu), \ \rho_{\mu}(f) \le 0\right\} = \alpha(\nu | \mu).
\end{align*}
Here $\E^{\bar{\nu}}$ denotes integration with respect to $\bar{\nu}$. In other words, acceptance consistency is equivalent to
\[
\alpha(\bar{\nu} | \bar{\mu}) \ge \alpha(\nu | \mu) + \E^{\bar{\nu}}\left[\alpha_{1}(\bar{\nu}) \right]
\]
holding for every $\bar{\nu} = \nu(dx)K^\nu_x(dy) \in \P_{\bar{\mu}}(E \times F)$. This differs from our definition of superadditivity only in the term $\E^{\bar{\nu}}[\alpha_1(\bar{\nu})]$. According to Lemma 4 of \cite{acciaio-penner-dynamic}, 
\begin{align*}
\E^{\bar{\nu}}\left[\alpha_{1}(\bar{\nu}) \right] &= \sup\left\{\int_{E \times F}f\,d\bar{\nu} : f \in L^\infty(E \times F,\F_2,\bar{\mu}), \ \rho_1(f) \le 0 \ a.s.\right\}.
\end{align*}
Note that $\rho_1(f) \le 0$ a.s. if and only if $\mu\{ x\in E :\rho_{K^\mu_x}(f(x,\cdot)) \le 0\} = 1$. Moreover, for any $f \in B(E \times F)$, if $g(x,y) := f(x,y) - \rho_{K^\mu_x}(f(x,\cdot))$ then $\rho_{K^\mu_x}(g(x,\cdot)) \le 0$ for all $x$. Thus
\[
\E^{\bar{\nu}}\left[\alpha_{1}(\bar{\nu}) \right] = \sup_{f \in B(E \times F)}\left\{\int f\,d\bar{\nu} - \int\nu(dx)\rho_{K^\mu_x}(f(x,\cdot))\right\}
\]
But according to Lemma \ref{le:integralconvexity}, this is in turn equal to
\[
\int_E\nu(dx)\alpha(K^\nu_x | K^\mu_x).
\]
In other words, Lemma \ref{le:integralconvexity} bridges our characterization of acceptance consistency with that of \cite[Theorem 27]{acciaio-penner-dynamic}, which we now see are equivalent.

\subsection{Weak time consistency}

A related notion of time consistency was studied by Weber in \cite{weber-distributioninvariant}. Namely, we say a law invariant risk measure $\rho$ is \emph{weakly acceptance consistent} if $\rho(X | \G) \le 0$ a.s. implies $\rho(X) \le 0$, for every $X \in L^\infty$ and every $\sigma$-field $\G \subset \F$. Similarly, $\rho$ is \emph{weakly rejection consistent} if $\rho(X | \G) > 0$ a.s. implies $\rho(X) > 0$.
The following result, due in large part to Weber \cite{weber-distributioninvariant}, characterizes weak time consistency in terms of measure acceptance sets as well as divergences.
Let us say that a set $\A \subset \P(\R)$ is \emph{locally measure convex} if for each $M > 0$ and each $Q \in \P(\A \cap \P[-M,M])$ the mean measure $\int_{\A \cap \P[-M,M]} Q(dm)m$ is in $\A$.

\begin{theorem} \label{th:weber}
Suppose $\alpha$ is a simplified divergence induced by a law invariant risk measure $\rho$ with acceptance set $\A$. The following are equivalent:
\begin{enumerate}
\item $\rho$ is weakly acceptance consistent.
\item $\A$ is locally measure convex.
\item For Polish spaces $E$ and $F$, and measures $\mu(dx)K^\mu_x(dy)$ and $\nu(dx)K^\nu_x(dy)$ in $\P(E \times F)$, we have
\[
\alpha\left(\nu(dx)K^\nu_x(dy) \ | \ \mu(dx)K^\mu_x(dy)\right) \ge \int\nu(dx)\alpha(K^\nu_x | K^\mu_x).
\]
\item For Polish spaces $E$ and $F$, and measures $\mu_1 \times \mu_2$ and $\nu(dx)K^\nu_x(dy)$ in $\P(E \times F)$, we have
\[
\alpha\left(\nu(dx)K^\nu_x(dy) \ | \ \mu_1\times\mu_2\right) \ge \int\nu(dx)\alpha(K^\nu_x | \mu_2).
\]
\end{enumerate}
Similarly, the same equivalences hold when ``acceptance'' is changed to ``rejection'', ``sub'' is changed to ``super'', and $\A$ is changed to $\A^c$. The equivalence of (1) and (2) holds without the assumption that $\alpha$ is simplified.
\end{theorem}
\begin{proof}
The implication $(1) \Leftrightarrow (2)$ in the following was first noticed by Weber \cite{weber-distributioninvariant}, and the rest is proven along the same lines as Theorem \ref{th:mainequivalence}, but we provide a sketch: Suppose first that (1) holds. Fix Polish spaces $E$ and $F$ and measures $\mu(dx)K^\mu_x(dy)$ and $\nu(dx)K^\nu_x(dy)$ in $\P(E \times F)$. It is easy to see (similar to Proposition \ref{pr:acceptanceconsistent-equivalences}) that weak acceptance consistency is equivalent to the following: for $f \in B(E \times F)$, 
\[
\mu\{\rho_{K^\mu_x}(f(x,\cdot)) \le 0\}=1 \quad \Rightarrow \quad \rho_{\bar{\mu}}(f) \le 0.
\]
Thus, by Lemma \ref{le:integralconvexity},
\begin{align*}
\int\nu(dx)\alpha(K^\nu_x|K^\mu_x) &= \sup_{f \in B(E \times F)}\left\{\int f\,d\bar{\nu} - \int\nu(dx)\rho_{K^\mu_x}(f(x,\cdot))\right\} \\
	&= \sup\left\{\int f\,d\bar{\nu} : f \in B(E \times F), \ \mu\{\rho_{K^\mu_x}(f(x,\cdot)) \le 0\}=1\right\} \\
	&\le \sup\left\{\int f\,d\bar{\nu} : f \in B(E \times F), \ \rho_{\bar{\mu}}(f) \le 0\right\} \\
	&= \alpha(\bar{\nu} | \bar{\mu}).
\end{align*}
This proves $(1) \Rightarrow (3)$. Since clearly (3) implies (4), let us finally show that (4) implies (1). Fix Polish spaces $E$ and $F$ and $\mu_1 \times \mu_2 \in \P(E \times F)$. As in the proof of Theorem \ref{th:mainequivalence}, the inequality of (4), combined with Lemma \ref{le:integralconvexity} and the order-reversing property of convex conjugation, implies the set inclusion
\[
\left\{f \in B(E \times F) : \rho_{\mu_1 \times \mu_2}(f) \le 0\right\} \supset \left\{f \in B(E \times F) : \mu_1\{\rho_{\mu_2}(f(x,\cdot)) \le 0\}=1\right\}.
\]
Again, it is easy to see (similar to Proposition \ref{pr:acceptanceconsistent-equivalences}) this implies weak acceptance consistency.
\end{proof}

\begin{remark}
In fact, for a measure acceptance set $\A$, local measure convexity is equivalent to (ordinary) convexity. Indeed, the set $\A \cap \P[-M,M]$ is weakly closed for each $M > 0$, which can be proven easily using the Fatou property \ref{th:jouini-touzi-schachermayer} and the Skorohod representation for weak convergence. It is well known that closed convex sets are measure convex, e.g. \cite[Corollary 1.2.4]{winkler-choquet}. The same equivalence may not hold for the complement $\A^c$, which is not closed.
\end{remark}

\subsection{More on shift-convexity} \label{se:shiftconvexity}
Let us recall our first interpretation of a shift-convex acceptance set: Suppose $X$ and $Y$ are risks, and $X$ is acceptable. Suppose that $Y$ is conditionally acceptable given $X$. Then shift-convexity means that $X+Y$ is itself acceptable. Proposition \ref{pr:shiftconvex-equivalence} below shows that this interpretation can be sharpened somewhat:
Suppose $X$ and $Y$ are risks, and $X$ is acceptable. Suppose that $X$ is $\G$-measurable for some $\sigma$-field $\G$, and suppose the risk $Y$ is conditionally acceptable given $\G$. Then shift-convexity implies that the risk $X+Y$ is acceptable as well.

\begin{proposition} \label{pr:shiftconvex-equivalence}
Suppose $\A$ is the measure acceptance set of a law invariant risk measure. Then $\A$ is shift-convex if and only if it satisfies the following property:
\begin{enumerate}
\item[(S)] For each $M > 0$ and each $\gamma \in \P(\R \times (\A \cap \P[-M,M]))$ with first marginal belonging to $\A$, 
it holds that 
\[
\int\gamma(dx,dm)m(\cdot - x) \in \A.
\]
\end{enumerate}
Similarly, $\A^c$ is shift-convex if and only if it satisfies property (S).
\end{proposition}
\begin{proof}
Suppose $\A$ satisfies property (S). Fix $\mu \in \A$ and a measurable map $\R \ni x \mapsto K_x \in \A \cap \P[-M,M]$. Define $\gamma \in \P(\R \times (\A \cap \P[-M,M]))$ now by
\[
\gamma(dx,dm) := \mu(dx)\delta_{K_x}(dm).
\]
Then clearly the first marginal of $\gamma$ is $\mu$, which belongs to $\A$. Moreover, $\gamma(\R \times (\A \cap \P[-M,M])) = 1$ since $K_x \in \A \cap \P[-M,M]$ for every $x$.
Thus, by property (S), the measure
\[
\int_\R\mu(dx)K_x(\cdot - x) = \int\gamma(dx,dm)m(\cdot - x)
\]
is in $\A$, which shows that $\A$ is shift-convex.

Conversely, suppose $\A$ is shift-convex. Then the corresponding law invariant risk measure $\rho$ is acceptance consistent. A fortiori, $\rho$ is weakly acceptance consistent and thus $\A$ is locally measure convex by Theorem \ref{th:weber}. Now fix $M > 0$ and $\gamma \in \P(\R \times (\A \cap \P[-M,M]))$ with first marginal $\mu$ belonging to $\A$. Disintegrate $\gamma$ to find a measurable map $\R \ni x \mapsto Q_x \in \P(\A \cap \P[-M,M])$ such that
\[
\gamma(dx,dm) = \mu(dx)Q_x(dm).
\]
For each $x$ define $K_x \in \P(\R)$ to be the mean measure of $Q_x$, i.e.
\[
K_x(\cdot) = \int Q_x(dm)m(\cdot).
\]
Since $Q_x(\A \cap \P[-M,M]) = 1$ for $\mu$-a.e. $x$, and since $\A$ is locally measure convex, it holds that $K_x \in \A$ for $\mu$-a.e. $x$. From partial shift-convexity we conclude that the measure
\begin{align*}
\int\gamma(dx,dm)m(\cdot - x) &= \int_\R\mu(dx)\int_{\A \cap \P[-M,M]}Q_x(dm)m(\cdot - x) = \int_\R\mu(dx)K_x(\cdot - x)
\end{align*}
is in $\A$, which proves property (S).
\end{proof}

\section{Further properties of divergences} \label{se:furtherproperties}

While every divergence is convex in its first argument by definition,
it is well known that relative entropy and also $f$-divergences are \emph{jointly} convex. It turns out that \emph{joint} convexity of a divergence is equivalent to \emph{concavity} of the corresponding law invariant risk measure on the level of distributions. To be clear, for a law invariant risk measure $\rho$, define the function $\tilde{\rho}$ on the set of probability measures on $\R$ with compact support by setting $\tilde{\rho}(P \circ X^{-1}) = \rho(X)$, for $X \in L^\infty$. The concavity of $\tilde{\rho}$ was studied recently by Acciaio and Svindland \cite{acciaio-svindland-concave}, who make a compelling case that concavity is much more common in spite of the convexity of $\rho$ on the level of random variables. Indeed, they show that $\rho(X) = \E X$ is the \emph{only} law invariant risk measure for which $\tilde{\rho}$ is convex. The entropic risk measure, for example, clearly has $\tilde{\rho}$ concave. Moreover, if $\rho$ is the optimized certainty equivalent corresponding to a function $\phi$, then the formula
\[
\tilde{\rho}(\mu) = \inf_{m \in \R}\left(\int\phi(m + x)\mu(dx) - m\right)
\]
shows that $\tilde{\rho}$ is concave.

\begin{proposition} \label{pr:jointconvexity}
Let $\rho$ be a law invariant risk measure with induced divergence $\alpha$. The following are equivalent:
\begin{enumerate}
\item $\alpha$ is jointly convex, in the sense that  $\alpha(\cdot|\cdot)$ is convex on $\P(E) \times \P(E)$ for each Polish space $E$.
\item For each Polish space $E$ and each $f \in B(E)$, the map $\mu \mapsto \rho_\mu(f)$ is concave.
\item $\tilde{\rho}$ is concave.
\end{enumerate}
\end{proposition}
\begin{proof}
($1 \Rightarrow 2$) Let $E$ be a Polish space and $f \in B(E)$. Fix $t \in (0,1)$ and $\mu_1,\mu_2 \in \P(E)$. Then (1) implies
\begin{align*}
\rho_{t\mu_1 + (1-t)\mu_2}(f) &= \sup_{\nu \in \P(E)}\left\{\int f\,d\nu - \alpha(\nu | t\mu_1 + (1-t)\mu_2)\right\} \\
	&\ge \sup_{\nu_1,\nu_2 \in \P(E)}\left\{t\int f\,d\nu_1 + (1-t)\int f\,d\nu_2 - \alpha(t\nu_1 + (1-t)\nu_2 | t\mu_1 + (1-t)\mu_2)\right\} \\
	&\ge \sup_{\nu_1,\nu_2 \in \P(E)}\left\{t\int f\,d\nu_1 + (1-t)\int f\,d\nu_2 - t\alpha(\nu_1 | \mu_1) + (1-t)\alpha(\nu_2 | \mu_2)\right\} \\
	&= t\rho_{\mu_1}(f) + (1-t)\rho_{\mu_2}(f).
\end{align*}

($2 \Rightarrow 1$) On the other hand, if $\nu_1,\nu_2 \in \P(E)$, then (2) implies
\begin{align*}
\alpha(t\nu_1 + (1-t)\nu_2 | t\mu_1 + (1-t)\mu_2) &= \sup_{f \in B(E)}\left\{t\int f\,d\nu_1 + (1-t)\int f\,d\nu_2 - \rho_{t\mu_1 + (1-t)\mu_2}(f)\right\} \\
	&\le \sup_{f \in B(E)}\left\{t\int f\,d\nu_1 + (1-t)\int f\,d\nu_2 - t\rho_{\mu_1}(f) + (1-t)\rho_{\mu_2}(f)\right\} \\
	&\le t\alpha(\nu_1 | \mu_1) + (1-t)\alpha(\nu_2 | \mu_2).
\end{align*}

($3 \Rightarrow 2$) This is immediate from the identity
\[
\rho_{t\mu_1 + (1-t)\mu_2}(f) = \tilde{\rho}\left(t\mu_1\circ f^{-1} + (1-t)\mu_2\circ f^{-1}\right).
\]

($2 \Rightarrow 3$) This is almost immediate from the above identity. Assume (2). Let $m_1,m_2 \in \P(\R)$ have compact support, and let $t \in (0,1)$. Then, letting $id$ denote the identity map on $\R$,
\[
\tilde{\rho}(tm_1 + (1-t)m_2) = \rho_{tm_1 + (1-t)m_2}(id) \ge t\rho_{m_1}(id) + (1-t)\rho_{m_2}(id) = t\tilde{\rho}(m_1) + (1-t)\tilde{\rho}(m_2).
\]
\end{proof}

Divergences are actually uniquely determined by their values for \emph{finite} spaces $E$, as is formalized in the following proposition. Building on the characterization of relative entropy in Corollary \ref{co:entropic} below, we could derive an even simpler characterization akin to those surveyed by Csisz\'ar \cite{csiszar2008axiomatic}, but this would lead us too far astray. 

\begin{proposition}
Suppose $\alpha$ is a simplified divergence.
For any Polish space $E$ and any $\mu,\nu \in \P(E)$, we have
\[
\alpha(\nu | \mu) = \sup\left\{\alpha(\nu \circ T^{-1}|\mu \circ T^{-1}) : T : E \rightarrow F \text{ measurable, } F \text{ finite}\right\}.
\]
\end{proposition}
\begin{proof}
The inequality $\ge$ follows immediately from the definition of a divergence. To prove the reverse inequality, note that it holds trivially if $E$ is finite. Generally, by Borel isomorphism (see \cite[Theorem 15.6]{kechris-settheory}), there exists a measurable function $S : E \rightarrow [0,1]$ with measurable inverse. Suppose we can prove that
\begin{align}
\alpha(\nu | \mu) = \sup\left\{\alpha(\nu \circ T^{-1}|\mu \circ T^{-1}) : T : [0,1] \rightarrow F \text{ measurable, } F \text{ finite}\right\}. \label{pf:finiteapprox1}
\end{align}
for all $\mu,\nu \in \P([0,1])$. Then, if $\mu,\nu \in \P(E)$, we use Proposition \ref{pr:informationinequality}  to conclude
\begin{align*}
\alpha(\nu | \mu) &= \alpha(\nu \circ S^{-1} | \mu \circ S^{-1}) \\
	&= \sup\left\{\alpha(\nu \circ (T \circ S)^{-1}|\mu \circ (T \circ S)^{-1}) : T : [0,1] \rightarrow F \text{ measurable, } F \text{ finite}\right\} \\
	&=\sup\left\{\alpha(\nu \circ T^{-1}|\mu \circ T^{-1}) : T : E \rightarrow F \text{ measurable, } F \text{ finite}\right\}.
\end{align*}
Indeed, this is true because every measurable map $T : E \rightarrow F$ can be written as $T' \circ S$, where $T' = T \circ S^{-1}$. Hence, we need only to prove \eqref{pf:finiteapprox1}.

Since $[0,1]$ is compact, for each $n$ we may find a measurable map $T_n : [0,1] \rightarrow [0,1]$ with finite range such that $|x - T_n(x)| \le 1/n$ for all $x \in [0,1]$.
Then $T_n$ converges uniformly to the identity. Since $\alpha$ is simplified, for a given $\epsilon > 0$ we may find a \emph{continuous} function $f$ on $[0,1]$ such that
\[
\alpha(\nu | \mu) \le \epsilon + \int f\,d\nu - \rho_\mu(f).
\]
Since $\rho_\mu$ is continuous in the supremum norm, and since $f \circ T_n \rightarrow f$ uniformly, we conclude that $\rho_\mu(f) = \lim_n\rho_\mu(f \circ T_n)$. Thus
\begin{align*}
\alpha(\nu | \mu) &\le \epsilon + \lim_{n\rightarrow\infty}\left(\int f \circ T_n\,d\nu - \rho_\mu(f \circ T_n)\right) \\
	&= \epsilon + \lim_{n\rightarrow\infty}\left(\int f\,d\nu \circ T_n^{-1} - \rho_{\mu \circ T_n^{-1}}(f)\right) \\
	&\le \epsilon + \liminf_{n\rightarrow\infty}\alpha(\nu \circ T_n^{-1} | \mu \circ T_n^{-1}).
\end{align*}
This is enough to complete the proof.
\end{proof}

Finally, let us mention a result of potential relevance in mathematical statistics, namely that sufficient statistics always attain equality in the inequality $\alpha(\nu \circ T^{-1} | \mu \circ T^{-1}) \le \alpha(\nu | \mu)$. 
The foundational paper \cite{kullback-leibler} of Kullback and Leibler proved this result for relative entropy, and Liese and Vadja \cite[Theorem 14]{liese2006divergences} treat the case of $f$-divergences.

\begin{proposition} \label{pr:sufficientstatistic}
Let $E$ be a Polish space and $\mu,\nu \in \P(E)$ with $\nu \ll \mu$. Suppose a measurable map $T : E \rightarrow F$ is sufficient for $\{\mu,\nu\}$, meaning that $d\nu/d\mu$ is $T$-measurable. Then, for any divergence $\alpha$, we have $\alpha(\nu \circ T^{-1} | \mu \circ T^{-1}) = \alpha(\nu | \mu)$. In particular, this holds if $T$ is a bijection with measurable inverse.
\end{proposition}
\begin{proof}
We use a more probabilistic notation for this proof, since we deal with conditional expectations.
By definition of a divergence, $\alpha(\nu \circ T^{-1} | \mu \circ T^{-1}) \le \alpha(\nu | \mu)$, so we must only prove the reverse inequality.
As is well known, sufficiency of $T$ implies easily that $\E^\nu[f \, | \, T] = \E^\mu[f \, | \, T]$ a.s. for each $f \in B(E)$; indeed, if $h \in B(E)$ is $T$-measurable, then
\[
\E^\nu[f h] = \E^\mu[f h d\nu/d\mu] = \E^\mu[\E^\mu[f \, | \, T] h d\nu/d\mu] = \E^\nu[\E^\mu[f \, | \, T] h].
\]
Because of Corollary 4.65 of \cite{follmer-schied-book}, we have $\rho_\mu(f) \ge \rho_\mu(\E^\mu[f \, | \, T])$. Thus
\begin{align*}
\alpha(\nu | \mu) &= \sup_{f \in B(E)}\left( \E^\nu[f] - \rho_\mu(f)\right) \\
	&\le \sup_{f \in B(E)}\left(\E^\nu[f] - \rho_\mu(\E^\mu[f \, | \, T])\right) \\
	&= \sup_{f \in B(E)}\left(\E^\nu[\E^\mu[f \, | \, T]] - \rho_\mu(\E^\mu[f \, | \, T])\right).
\end{align*}
Every $T$-measurable function on $E$ may be written as $g \circ T$ for some measurable function $g$, and thus
\[
\alpha(\nu | \mu) \le \sup_{g \in B(F)}\left(\E^\nu[g \circ T] - \rho_\mu(g \circ T)\right) = \sup_{g \in B(F)}\left(\E^{\nu \circ T^{-1}}[g] - \rho_{\mu\circ T^{-1}}(g)\right) = \alpha(\nu \circ T^{-1} | \mu \circ T^{-1}).
\]
\end{proof}

\begin{remark}
Proposition \ref{pr:sufficientstatistic} raises a natural question: Does the converse hold? That is, does the equality $\alpha(\nu \circ T^{-1} | \mu \circ T^{-1}) = \alpha(\nu | \mu)$ imply that $T$ is sufficient for $(\mu,\nu)$? This does not hold for all divergences, but it does when $\alpha$ is relative entropy, as was observed first by Kullback and Leibler \cite{kullback-leibler}. Liese and Vadja \cite{liese2006divergences} show that this converse holds for many (but not all) $f$-divergences. This characterization leads to useful tests for sufficiency, as is explained in both of these papers \cite{kullback-leibler,liese2006divergences}.
\end{remark}

\section{Examples} \label{se:examples}

Before we discuss some common law invariant risk measures, recall that our sign convention is not the usual one. Namely, $\rho$ is increasing, not decreasing. More precisely, if $\rho$ is a risk measure according to our definition, the map $ X \mapsto \rho(-X)$ is what is more often called a risk measure, as in \cite{follmer-schied-book}.

\subsection{Shortfall risk measures} \label{se:shortfall}

Shortfall risk measures, introduced by F\"ollmer and Schied \cite{follmer-schied-convex}, are of the form
\[
\rho(X) = \inf\{c \in \R : \E[\ell(X-c)] \le 1\},
\]
where $\ell$ is a \emph{loss function}, defined as follows:

\begin{definition} \label{def:lossfunction}
A \emph{loss function} is a convex and nondecreasing function $\ell : \R \rightarrow \R$ satisfying $\ell(0) = 1 < \ell(x)$ for all $x > 0$. 
\end{definition}

Of course, the induced family of risk measures is
\begin{align}
\rho_\mu(f) = \inf\{c \in \R : \int_E\ell(f(x)-c)\mu(dx) \le 1\}. \label{def:shortfall}
\end{align}
Note that by continuity of $\ell$ and monotone convergence, the infimum is always attained. In particular,
\[
\int_E\ell(f(x)-\rho_\mu(f))\mu(dx) \le 1.
\]
According to the \cite[Theorem 4.115]{follmer-schied-book} the induced divergence is
\begin{align}
\alpha(\nu | \mu) = \inf_{t > 0}\frac{1}{t}\left(1 + \int_E\ell^*\left(t\frac{d\nu}{d\mu}\right)d\mu\right), \text{ for } \nu \ll \mu, \label{def:shortfallentropy}
\end{align}
where $\ell^*(x) = \sup_{y \in \R}(xy - \ell(y))$ is the convex conjugate.
It is known that shortfall risk measures are Lebesgue continuous in the sense of Definition \ref{def:lebesgue-continuous} \cite[Proposition 4.113 and Exercise 4.2.2]{follmer-schied-book}. Hence, by Theorem \ref{th:tight-lowersemicontinuity}, $\alpha$ is jointly weakly lower semicontinuous and simplified. 
Let us now determine when a shortfall risk measure is superadditive. We say that a nonnegative function $\ell : \R \rightarrow [0,\infty]$ is \emph{log-subadditive} (resp. log-superadditive) if $\ell(x+y) \le \ell(x)\ell(y)$ (resp. $\ge$) for all $x,y \in \R$.

\begin{proposition} \label{pr:shortfalltensorization}
Let $\ell$ be a loss function and $\alpha$ the corresponding divergence defined in \eqref{def:shortfallentropy}. If $\ell$ is log-subadditive (resp. log-superadditive) then $\alpha$ is superadditive (resp. subadditive), or equivalently $\rho$ is acceptance consistent (resp. rejection consistent)
\end{proposition}
\begin{proof}
Assume $\ell$ is log-subadditive. With Theorem \ref{th:mainequivalence} in mind, we will show that the following set is shift-convex:
\[
\A := \left\{m \in \bigcup_{M > 0}\P[-M,M] : \int\ell\,dm \le 1\right\}.
\]
Fix $\mu \in \A$, $M > 0$, and a measurable map $\R \ni x \mapsto K_x \in \A \cap \P[-M,M]$. Then $\int_\R \ell\,dK_x \le 1$ for all $x$ and also $\int_\R\ell\,d\mu \le 1$. Thus
\begin{align*}
\int_\R\mu(dx)\int_\R K_x(dy-x)\ell(y) &= \int_\R\mu(dx)\int_\R K_x(dy)\ell(y+x) \\
	&\le \int_\R\mu(dx)\ell(x)\int_\R K_x(dy)\ell(y) \\
	&\le 1,
\end{align*}
and it follows that $\int_\R\mu(dx)K_x(\cdot - x)$ is in $\A$.
\end{proof}

\begin{remark} \label{re:notmanylogsubadditive}
It is difficult to construct interesting examples of log-subadditive functions, beyond the obvious case of $\ell(x)=e^{\eta x}$ for $\eta > 0$. Note that $\ell(x) = 0$ for some $x$ precludes log-subadditivity, since it implies $\ell(y)\le\ell(y-x)\ell(x)=0$ for all $y$. The function $\log \ell$ must be nondecreasing and subadditive on $\R$ and equal to zero at zero, and moreover the exponential of this function must be convex. The only other examples we found are all of the restrictive form $\ell(x)=e^{F(x)}$ for nondecreasing functions $F$ with $F'(0) > 0$ and with at most linear growth, i.e., $F(x) \le c_1x$ for all $x \ge 0$ and $F(x) \ge c_2x$ for $x \le 0$, for $c_1,c_2 > 0$.
\end{remark}

\subsection{Optimized certainty equivalent}
An optimized certainty equivalent, as introduced by Ben-Tal and Teboulle \cite{bental-teboulle-1986,bental-teboulle-2007}, is of the form
\[
\rho(X) := \inf_{m \in \R}\left(\E[\phi(m+X)] - m\right),
\]
where $\phi : \R \rightarrow \R$ is convex and nondecreasing, with $\phi^*(1) = \sup_{x \in \R}(x - \phi(x)) = 0$.
Of course, the induced family of risk measures is
\[
\rho_\mu(f) := \inf_{m \in \R}\left(\int_E\phi(m+f(x))\mu(dx) - m\right).
\]
The corresponding divergence is the $\phi^*$-divergence,
\[
\alpha(\nu|\mu) = \int\phi^*\left(\frac{d\nu}{d\mu}\right)d\mu, \text{ for } \nu \ll \mu.
\]
As we saw in the discussion preceding Proposition \ref{pr:jointconvexity}, an optimized certainty equivalent always satisfies the concavity condition of Proposition \ref{pr:jointconvexity}, and this provides an alternative proof of the well known joint convexity of $\alpha$. It is also known that $\alpha$ is jointly weakly lower semicontinuous, which we confirm using Theorem \ref{th:tight-lowersemicontinuity} before addressing time consistency and additivity properties.

\begin{proposition}
Every optimized certainty equivalent is weakly lower semicontinuous. In particular, $\alpha$ is simplified.
\end{proposition}
\begin{proof}
The second claim follows from the first by Theorem \ref{th:tight-lowersemicontinuity}.
According to Theorem \ref{th:tight-lowersemicontinuity}, it suffices to show Lebesgue continuity: If $X_n,X \in L^\infty$ with $X_n \downarrow X$ a.s., we must show that $\rho(X_n) \downarrow \rho(X)$. Let $\epsilon > 0$, and find $m \in \R$ such that
\[
\E[\phi(m+X)] - m \le \rho(X) + \epsilon.
\]
By monotone convergence, $\E[\phi(m + X_n)] \downarrow \E[\phi(m+X)]$. Thus, for sufficiently large $n$,
\[
\rho(X_n) \le \E[\phi(m + X_n)] - m \le \E[\phi(m+X)] - m + \epsilon \le \rho(X) + 2\epsilon.
\]
\end{proof}

\begin{theorem}
Suppose $\phi$ satisfies $y\phi^*(x) + x\phi^*(y) \le \phi^*(xy)$ (resp. $\ge$) for all $x,y \ge 0$. Then $\alpha$ is superadditive (resp. subadditive), or equivalently $\rho$ is acceptance consistent (resp. rejection consistent).
\end{theorem}
\begin{proof}
We treat the superadditive case, as the subadditive case is proven similarly.
Suppose $E$ and $F$ are Polish spaces, and let $\bar{\mu} = \mu(dx)K^\mu_x(dy)$ and $\bar{\nu} = \nu(dx)K^\nu_x(dy)$ be probability measures on $E \times F$. 
Simply use the definition of $\alpha$:
\begin{align*}
\alpha(\bar{\nu} | \bar{\mu}) &= \int_{E \times F}\phi^*\left(\frac{d\bar{\nu}}{d\bar{\mu}}\right)d\mu \\
	&= \int_E\mu(dx)\int_F K^\mu_x(dy)\phi^*\left(\frac{d\nu}{d\mu}(x)\frac{dK^\nu_x}{dK^\mu_x}(y)\right) \\
	&\ge \int_E\mu(dx)\int_FK^\mu_x(dy)\left[\frac{dK^\nu_x}{dK^\mu_x}(y)\phi^*\left(\frac{d\nu}{d\mu}(x)\right) + \frac{d\nu}{d\mu}(x)\phi^*\left(\frac{dK^\nu_x}{dK^\mu_x}(y)\right) \right] \\
	 &= \int_E\phi^*\left(\frac{d\nu}{d\mu}\right)d\mu + \int_E\nu(dx)\int_FK^\mu_x(dy)\phi^*\left(\frac{dK^\nu_x}{dK^\mu_x}(y)\right) \\
	&= \alpha(\nu | \mu) + \int_E\nu(dx)\alpha(K^\nu_x|K^\mu_x).
\end{align*}
\end{proof}

\begin{remark}
Of course, the relationship $\phi^*(xy) = x\phi^*(y) + y\phi^*(x)$ is satisfied by $\phi^*(x) = x\log x$, the conjugate of which (assuming $\phi^* = +\infty$ on the negative half-line) is $\phi(x) = e^{x-1}$. More generally, suppose $\ell$ is a strictly increasing log-subadditive loss function. Then $\ell^{-1}(xy) \ge \ell^{-1}(x) + \ell^{-1}(y)$ for $x,y > 0$, and so $\phi^*(x) := x\ell^{-1}(x)$ satisfies $\phi^*(xy) \ge x\phi^*(y) + y\phi^*(x)$. As was discussed in Remark \ref{re:notmanylogsubadditive}, there are not many such functions.
\end{remark}

\subsection{Coherent risk measures}

A risk measure is called \emph{coherent} if $\rho(\lambda X) = \lambda\rho(X)$ for all $X \in L^\infty$ and $\lambda \ge 0$. A coherent law invariant risk measure admits a representation
\[
\rho(X) = \sup_{Q \in \Q}\E^Q[X],
\]
where $\Q \subset \P_P(\Omega)$ is closed convex.  Law invariance simply means that if $Q \in \Q$ and $Q'\in \P_P(\Omega)$ have the same density law $ P \circ (dQ/dP)^{-1} = P \circ (dQ'/dP)^{-1}$, then $Q'$ must also be in $\Q$; indeed, this follows easily from the Kusuoka representation (see \cite{kusuoka2001law,jouini-touzi-schachermayer}). For a Polish space $E$ and $\mu \in \P(E)$, note that
\[
\rho_\mu(f) = \sup_{\eta \in \Q[\mu]}\int f\,d\eta,
\]
where
\[
\Q[\mu] := \left\{Q \circ X^{-1} : X \in L^0(\Omega;E), \ Q \in \Q, \ P \circ X^{-1} = \mu \right\} \subset \P_\mu(E).
\]
Here $L^0(\Omega;E)$ denotes the set of measurable functions from $\Omega$ to $E$.
Let us see when $\rho$ is acceptance consistent by using Proposition \ref{pr:acceptanceconsistent-equivalences}(5). Let $E$ and $F$ be Polish spaces, let $\mu_1 \in \P(E)$, let $\mu_2 \in \P(F)$, and let $f \in B(E \times F)$. Then, if $X$ denotes the identity map on $E$,
\begin{align}
\rho_{\mu_1}(\rho_{\mu_2}(f(x,\cdot))|_{x = X}) &= \sup_{\eta \in \Q[\mu]}\int_E\eta(dx)\rho_{\mu_2}(f(x,\cdot)) \nonumber \\
	&= \sup_{\eta \in \Q[\mu_1]}\int_E\eta(dx)\sup_{\eta' \in \Q[\mu_2]}\int_E\eta'(dy)f(x,y) \nonumber \\ 
	&= \sup_{\eta \in \Q[\mu_1,\mu_2]}\int_{E \times F}f\,d\eta, \label{pf:coherent}
\end{align}
where we define 
\[
\Q[\mu_1,\mu_2] := \left\{m(dx)K^m_x(dy) \in \P(E \times F) : m \in \Q[\mu_1], \ K^m_x \in \Q[\mu_2] \text{ for all } x\right\}.
\]
Indeed, the last line of \eqref{pf:coherent} follows from a well known measurable selection argument \cite[Proposition 7.50]{bertsekasshreve}. Thus, $\rho$ is acceptance consistent if and only if 
\[
\sup_{\eta \in \Q[\mu_1 \times \mu_2]}\int_{E \times F}f\,d\eta \le \sup_{\eta \in \Q[\mu_1,\mu_2]}\int_{E \times F}f\,d\eta,
\]
for all $f \in B(E \times F)$, since the left-hand side is exactly $\rho_{\mu_1 \times \mu_2}(f)$. But this is equivalent to the closed convex hull of $\Q[\mu_1,\mu_2]$ containing $\Q[\mu_1 \times \mu_2]$. For this to hold for every pair $\mu_1,\mu_2$ is a very stringent requirement, and we were unable to clarify it any further. It holds in the extreme cases, when $\Q$ is a singleton or $\Q = \P_P(\Omega)$.
Note that the above discussion is equally valid for the robust entropic risk measure
\[
\rho(X) = \sup_{Q \in \Q}c^{-1}\log\E^Q[e^{c X}], \ c > 0.
\]

\appendix

\section{Proof of Lemma \ref{le:integralconvexity}} \label{se:pf:integralconvexity}
Keep the notation of Lemma \ref{le:integralconvexity}.
For each $x \in E$ we have
\[
\alpha(K^\nu_x|K^\mu_x) = \sup_{f \in B(F)}\left(\int_F f\,dK^\nu_x  - \rho_{K^\mu_x}(f)\right),
\]
and this clearly implies (recalling that $\bar{\nu} := \nu(dx)K^\nu_x(dy)$) that
\begin{align*}
\int_E\nu(dx)\alpha(K^\nu_x|K^\mu_x) \ge \sup_{f \in B(E \times F)}\left\{\int_{E \times F} f\,d\bar{\nu} - \int_E\nu(dx)\rho_{K^\mu_x}(f(x,\cdot))\right\}.
\end{align*}
The rest of the proof is devoted to establishing the trickier inequality:
\begin{align}
\int_E\nu(dx)\alpha(K^\nu_x|K^\mu_x) \le \sup_{f \in B(E \times F)}\left\{\int_{E \times F} f\,d\bar{\nu} - \int_E\nu(dx)\rho_{K^\mu_x}(f(x,\cdot))\right\}. \label{pf:integralconvexity1}
\end{align}

\textbf{Proof for $F=[0,1]$.}
Assume first that $F = [0,1]$. 
Since $\alpha$ is simplified, for each $x \in E$ we have
\[
\alpha(K^\nu_x|K^\mu_x) = \sup_{f \in C([0,1])}\left(\int_{[0,1]} f\,dK^\nu_x  - \rho_{K^\mu_x}(f)\right).
\]
Now note that for each $f \in C([0,1])$, the function of $x$ inside the supremum on the right-hand side above is measurable. Indeed, we showed in Lemma \ref{le:rhofamilyregularity} that $\mu \mapsto \rho_\mu(f)$ is measurable (actually lower semicontinuous) when $f$ is continuous and bounded.
Since $C([0,1])$ equipped with the supremum norm is a Polish space, for a fixed $\epsilon > 0$ we may find (by \cite[Proposition 7.50]{bertsekasshreve}) a universally measurable map $E \ni x \mapsto g_x \in C([0,1])$ such that, for each $x\in E$,
\begin{align}
\int_{[0,1]} g_x\,dK^\nu_x  - \rho_{K^\mu_x}(g_x) \ge \begin{cases}
\alpha(K^\nu_x|K^\mu_x) - \epsilon &\text{if } \alpha(K^\nu_x|K^\mu_x) < \infty \\
1/\epsilon &\text{if } \alpha(K^\nu_x|K^\mu_x) = \infty.
\end{cases} \label{pf:integralconvexity2}
\end{align}
We may find a Borel measurable map which agrees $\nu$-a.e. with $x \mapsto g_x$, and we abuse notation by denoting this again by $x \mapsto g_x$.
By Lusin's theorem, for each $\delta > 0$ there exists a compact set $S_\delta \subset E$ such that $\nu(S_\delta^c) \le \delta$ and the restriction $S_\delta \ni x \mapsto g_x \in C([0,1])$ is continuous. Without loss of generality, assume that $S_\delta \supset S_{\delta'}$ whenever $\delta < \delta'$. Define $g^\delta_x := g_x$ for $x \in S_\delta$ and $g^\delta_x := 0$ (the zero function) for $x \notin S_\delta$, and note that $g^\delta$ is a \emph{bounded} measurable function from $E$ to $C([0,1])$.
Define $g^\delta : E \times [0,1] \rightarrow \R$ by $g^\delta(x,y)=g^\delta_x(y)$, and note that $g^\delta$ is jointly measurable (thanks to Theorem 4.55 and Lemma 4.51 of \cite{aliprantisborder}) and bounded.

\textit{Case 1.}
Suppose first that $\int_E\nu(dx)\alpha(K^\nu_x|K^\mu_x) < \infty$. 
Then, for sufficiently small $\delta$, we have
\[
\int_{S_\delta^c}\nu(dx)\alpha(K^\nu_x|K^\mu_x) \le \epsilon.
\]
Then
\begin{align*}
\int_E\nu(dx)\alpha(K^\nu_x|K^\mu_x) &\le \epsilon + \int_{S_\delta}\nu(dx)\alpha(K^\nu_x|K^\mu_x) \\
	&\le 2\epsilon + \int_{S_\delta}\nu(dx)\left[\int_{[0,1]} g_x\,dK^\nu_x  - \rho_{K^\mu_x}(g_x)\right] \\
	&= 2\epsilon + \int_E\nu(dx)\left[\int_{[0,1]} g^\delta_x\,dK^\nu_x  - \rho_{K^\mu_x}(g^\delta_x)\right] \\
	&= 2\epsilon + \int_{E \times [0,1]} g^\delta\,d\bar{\nu} - \int_E\nu(dx)\rho_{K^\mu_x}(g^\delta(x,\cdot)) \\
	&\le 2\epsilon + \sup_{f \in B(E \times [0,1])}\left\{\int_{E \times [0,1]} f\,d\bar{\nu} - \int_E\nu(dx)\rho_{K^\mu_x}(f(x,\cdot))\right\}.
\end{align*}
The second line used \eqref{pf:integralconvexity2}, and the third used the definition of $g^\delta$.
Since $\epsilon > 0$ was arbitrary, we obtain  \eqref{pf:integralconvexity1}.

\textit{Case 2.}
Suppose the set $I := \{x \in E : \alpha(K^\nu_x|K^\mu_x) = \infty\}$ has $\nu(I) > 0$, so that $\int_E\nu(dx)\alpha(K^\nu_x|K^\mu_x) = \infty$. Then for small enough $\delta$ we have $\nu(S_\delta \cap I) \ge \nu(I)/2 > 0$. Then, again using \eqref{pf:integralconvexity2} and the definition of $g^\delta$,
\begin{align}
0 < \frac{\nu(I)}{2\epsilon} &\le \frac{\nu(S_\delta \cap I)}{\epsilon} = \int_{S_\delta \cap I}\epsilon^{-1}\nu(dx) \nonumber \\
	&\le \int_{S_\delta \cap I}\nu(dx)\left[\int_{[0,1]} g_x\,dK^\nu_x  - \rho_{K^\mu_x}(g_x)\right] \nonumber \\
	&\le \epsilon\nu(I^c) + \int_E\nu(dx)\left[\int_{[0,1]} g^\delta_x\,dK^\nu_x  - \rho_{K^\mu_x}(g^\delta_x)\right] \nonumber \\
	&\le \epsilon\nu(I^c) + \sup_{f \in B(E \times [0,1])}\left\{\int_{E \times [0,1]} f\,d\bar{\nu} - \int_E\nu(dx)\rho_{K^\mu_x}(f(x,\cdot))\right\}. \label{pf:integralconvexity3}
\end{align}
The second to last inequality followed from the fact that
\[
\int_{I^c}\nu(dx)\int_{[0,1]} g^\delta_x\,dK^\nu_x  - \rho_{K^\mu_x}(g^\delta_x)  \ge -\epsilon\nu(I^c),
\]
by definition of $g^\delta$.
Send $\epsilon \downarrow 0$ in \eqref{pf:integralconvexity3} to conclude that the right-hand side of \eqref{pf:integralconvexity1} is also infinite.

\textit{Case 3.}
Finally, suppose $\int_E\nu(dx)\alpha(K^\nu_x|K^\mu_x) = \infty$ but $\nu(I) = 0$, with $I$ defined as in Case 2.
Since $\alpha \ge 0$, the monotone convergence theorem implies
\[
\lim_{\delta \downarrow 0}\int_{S_\delta}\nu(dx)\alpha(K^\nu_x|K^\mu_x) = \infty.
\]
Fix $M > 0$. Find $\delta$ such that $\int_{S_\delta}\nu(dx)\alpha(K^\nu_x|K^\mu_x) \ge M$. Then, again using \eqref{pf:integralconvexity2}, $\nu(I)=0$, and the definition of $g^\delta$,
\begin{align*}
M &\le \epsilon + \int_{S_\delta}\nu(dx)\left[\int_{[0,1]} g_x\,dK^\nu_x  - \rho_{K^\mu_x}(g_x)\right] \\
	&= \epsilon + \int_{E}\nu(dx)\left[\int_{[0,1]} g^\delta_x\,dK^\nu_x  - \rho_{K^\mu_x}(g^\delta_x)\right] \\
	&\le \epsilon + \sup_{f \in B(E \times [0,1])}\left\{\int_{E \times [0,1]} f\,d\bar{\nu} - \int_E\nu(dx)\rho_{K^\mu_x}(f(x,\cdot))\right\}
\end{align*}
Since $M > 0$ was arbitrary, this shows that the right-hand side is infinite, verifying \eqref{pf:integralconvexity1}.

\textbf{Proof for general $F$.}
Finally, we remove the assumption that $F = [0,1]$, merely assuming now that it is Polish. By Borel isomorphism (see \cite[Theorem 15.6]{kechris-settheory}) we may find a measurable bijection $T : F \rightarrow [0,1]$ with measurable inverse. Note that Proposition \ref{pr:informationinequality} yields $\alpha(K^\nu_x|K^\mu_x) = \alpha(K^\nu_x \circ T^{-1}|K^\mu_x \circ T^{-1})$. Apply the above result to prove \eqref{pf:integralconvexity1}:
\begin{align*}
\int_E\nu(dx)\alpha(K^\nu_x|K^\mu_x) &= \int_E\nu(dx)\alpha(K^\nu_x \circ T^{-1}|K^\mu_x \circ T^{-1}) \\
	&= \sup_{f \in B(E \times [0,1])}\int_E\nu(dx)\left[\int_{[0,1]}K^\nu_x \circ T^{-1}(dy)f(x,y) -\rho_{K^\mu_x \circ T^{-1}}(f(x,\cdot))\right] \\
	&= \sup_{f \in B(E \times [0,1])}\int_E\nu(dx)\left[\int_FK^\nu_x(dy)f(x,T(y)) -\rho_{K^\mu_x}(f(x,T(\cdot)))\right] \\
	&\le \sup_{f \in B(E \times F)}\int_E\nu(dx)\left[\int_FK^\nu_x(dy)f(x,y) -\rho_{K^\mu_x}(f(x,\cdot))\right].
\end{align*}

\subsection*{Acknowledgements}
Thanks are due to Igor Cialenco, Andrew Papanicolaou, and especially Stefan Weber for valuable discussions and references, as well as Kavita Ramanan for careful feedback on an early draft. Part of this research was performed while the author was visiting the Institute for Pure and Applied Mathematics (IPAM), which is supported by the National Science Foundation.

\bibliographystyle{amsplain}
\bibliography{../riskmeasures-bib}

\end{document}